\documentclass[%
superscriptaddress,
 amsmath,amssymb,
 aps,
pra,
twocolumn
]{revtex4-2}

\usepackage{amsthm,xcolor}
\DeclareMathOperator{\Tr}{Tr}
\DeclareMathOperator{\Cov}{\mathbf{Cov}}
\DeclareMathOperator{\Covel}{Cov}
\DeclareMathOperator{\Prec}{\mathbf{Prec}}
\DeclareMathOperator{\Precel}{Prec}
\DeclareMathOperator{\Var}{Var}

\DeclareMathOperator{\Real}{Re}
\DeclareMathAlphabet{\mathbbmsl}{U}{bbm}{m}{sl}
\newcommand{\dirprod}[1]{\underline{#1}}

\usepackage{graphicx}
\usepackage{bm}
\usepackage{hyperref}

\begin{document}

\title{Weak Convexity of Fisher Information Matrix and Superresolved Localization of Blinking Sources of Light}

\author{Dmitri B. Horoshko}\email{dmitri.horoshko@uni-ulm.de}
\affiliation{Institut f\"ur Quantenoptik, Universit\"at Ulm, Ulm D-89073, Germany}
\author{Alexander B. Mikhalychev}
\affiliation{Atomicus GmbH, Amalienbadstr. 41C, Karlsruhe, 76227, Germany}
\author{Fedor Jelezko}
\affiliation{Institut f\"ur Quantenoptik, Universit\"at Ulm, Ulm D-89073, Germany}
\author{Polina P. Kuzhir}
\affiliation{University of Eastern Finland, Yliopistokatu 7, Joensuu, 80101, Finland}

\date{\today}

\begin{abstract}
A group of techniques known by the general name of single-molecule localization microscopy reaches a nanometer-scale spatial resolution of point light emitters, well below the diffraction limit of the traditional microscopy. The key feature of these techniques is blinking, alternation of bright and dark states, of each emitter so that no more than one emitter is bright within the width of the point-spread function of the microscope during a time sufficient for its localization. We give a formulation of the optical part of these techniques in terms of quantum metrology, where the limit of precision is determined by the Fisher information on the emitters positions contained in the measurement data. We show that the advantage in resolution provided by making the emitters blink is a consequence of the fundamental property of Fisher information, its convexity. In particular, we prove the weak matrix convexity and the trace convexity of the Fisher information matrix -- two fundamental results in the multiparameter estimation theory. We show also that the advantage in information is inherent to the quantum state of light itself before the measurement and is related to the convexity of the quantum Fisher information matrix. 
\end{abstract}

\maketitle

\section{Introduction}
\label{sec:Intro}

The lateral resolution of a conventional far-field microscope is fundamentally restricted by the diffraction limit, equal to about half the wavelength of the illuminating light \cite{Born&Wolf}, which is above 200 nm in the visible range. In the last 30 years, numerous methods for surpassing the diffraction limit have been developed and are known by the general name of superresolution microscopy \cite{Vangindertael18,Schermelleh19,Defienne24}. One of the most successful directions in this area of research is related to subdiffraction imaging of point-like incoherent sources of light represented by fluorescent molecules \cite{Deschout14,Lelek21,Fazel24}, achieving nanometer-scale lateral resolution and distinguished by the Nobel Prize in Chemistry in 2014 'for the development of superresolved fluorescence microscopy.' An important subdirection is represented by the techniques of stochastically activated (blinking) fluorophores such as photoactivated localization microscopy (PALM) \cite{Betzig06}, stochastic optical reconstruction microscopy (STORM) \cite{Rust06}, points accumulation for imaging in nanoscale topography (PAINT) \cite{Sharonov06}, and their variants, know by the general name of single-molecule localization microscopy (SMLM) and having found numerous applications in cell biology \cite{Lelek21}. The SMLM techniques have also been applied to color centers in diamond providing subdiffraction-limit sensing of magnetic or electric fields \cite{Qu13,Pfender14}.

The general principle of SMLM consists in making the fluorophores blink, by a mechanism specific to each technique, so that not more than one source is present within the width of the point-spread function (PSF) of the microscope during a time sufficient for its precise localization. The latter is achieved by a simple arithmetic mean of the photodetection positions or by taking an analytical model for PSF and estimating the position of its peak by least squares or maximum likelihood estimation \cite{Small14}. The standard deviation of the source position scales as $1/\sqrt{N}$ with the number $N$ of detected photons \cite{Thompson02}, which allowed one to localize isolated fluorophores with an uncertainty of about 1 nm by detecting $N\approx10^4$ photons for each of them by a technique known as fluorescence imaging with one-nanometer accuracy (FIONA) \cite{Yildiz03}. Similarly, the positions of multiple sources within PSF can be found by modeling the intensity distribution on the microscope camera by a sum of PSFs and estimating the positions of its peaks, albeit with a much larger uncertainty \cite{Lidke05,Huang11}.

The $1/\sqrt{N}$ scaling of the position standard deviation is easily understood in the framework of mathematical statistics, namely, its branch known as point estimation theory \cite{Cox&Hinkley,Lehmann-Casella}. The central relation of this theory is the Cram\'er-Rao bound for the variance of any unbiased estimator of the parameter $x$:
\begin{equation}
\left(\Delta x\right)^2 \ge 1/\dirprod{F},
\end{equation}
where $\dirprod{F}$ is the Fisher information on this parameter contained in the measurement data. Here and below underlined letters denote quantities related to the entire measurement session typically consisting of many photodetections. In the case of localization of a single emitter, the subsequent photodetections are statistically independent, and the total Fisher information carried by $N$ photons is simply $\dirprod{F}=NF_\text{1phot}$, where $F_\text{1phot}$ is the Fisher information on the emitter position carried by one photon. The Cram\'er-Rao bound can be asymptotically reached in the high $N$ limit when the estimator is given by the maximum likelihood estimation \cite{Cox&Hinkley,Lehmann-Casella}. Thus, the $1/\sqrt{N}$ scaling of the position standard deviation is a mere consequence of the \emph{additivity} of Fisher information with respect to independent observations. The single-photon Fisher information can be calculated for various models of the PSF and the accompanying noise \cite{Ram06,Chao16} providing the Cram\'er-Rao bound for any given $N$. The variance of the estimated emitter position was shown to approach the Cram\'er-Rao bound for $N=10^2 - 10^4$ photons allowing thus for superresolution. 

The same argument can be applied to determining the separation of two closely spaced emitters: the total Fisher information carried by $N$ photons on the separation is $\dirprod{F}=NF_\text{sep,1phot}$, where $F_\text{sep,1phot}$ is the information on the separation carried by one photon. This information tends to zero together with the separation of the emitters \cite{Tsang16}, however, for whatever small $F_\text{sep,1phot}$, one can make the total information arbitrarily large by collecting sufficiently large number of photons and thus determine the separation with arbitrary precision. In practice, a fluorophore emits a limited number of photons before photobleaching, and, therefore, individual localization of blinking emitters gives much better results than a simultaneous localization of multiple cofluorescent emitters within the width of PSF.  The precision gain attained by making the sources blink is the key feature of SMLM.  A question arises: \emph{Is it a general law of Nature that blinking emitters send more information on their locations than cofluorescent ones?}

In this paper, we give a positive answer to this question, showing that the informational advantage gained by blinking is a consequence of another fundamental property of Fisher information, its \emph{convexity}. We consider the task of localization of an arbitrary number of light sources as a multiparameter estimation problem and prove that, in a scenario where all sources emit the same number of photons and microscopy satisfies the conditions of translation and rotation invariance, the Fisher information matrix possesses the property of \emph{weak matrix convexity}, which results in a higher localization efficiency for blinking sources. We also prove, that in the general case of arbitrary numbers of emitted photons, the Fisher information matrix possesses the property of \emph{trace convexity}, which leads to a higher localization efficiency of blinking sources when it is measured in terms of the new measure, which we call ``average eigenparameter precision.''   

The paper is structured as follows. In Sec.~\ref{sec:Formalism} we recast the formalism of quantum metrology \cite{Giovannetti06,Pirandola18,Albarelli20,Barbieri22} in the context of localization of dim incoherent sources of light \cite{Tsang16,Tsang17,Rehacek17,Bisketzi19,Hervas24}, extending a more traditional approach of classical metrology typically used in this context  \cite{Ram06,Chao16,mikhalychev2019efficiently,Vlasenko20,Kurdzialek21,mikhalychev2021lost}, and introduce three measures of localization efficiency. In Sec.~\ref{sec:Additivity}, we prove the fundamental properties of the Fisher information matrix mentioned above. In Sec.~\ref{sec:Localization}, we compare two scenarios of cofluorescent and blinking sources and show that the latter always has an informational advantage. We illustrate the developed theory by an example of two incoherent sources in Sec.~\ref{sec:Example} and extend the results to the quantum Fisher information matrix in Sec.~\ref{sec:Extension}. The obtained results are summarized in the Conclusion.

\section{Quantum metrology formalism for localization of point sources of light \label{sec:Formalism}}

\subsection{Cram\'er-Rao bound \label{sec:CRB}}
The objective of quantum metrology is inference of an unknown parameter $\theta$, which may be a vector $\theta=\{\theta_1,...,\theta_M\}$, encoded into the state of the object $\rho_O(\theta)$, which is a density operator in the Hilbert space of the object. For this purpose a generalized quantum measurement of the object is repeated $N$ times, the object being prepared in the same state $\rho_O(\theta)$ for every trial, so that the state of all $N$ object samples is $\dirprod{\rho}_O(\theta)=\rho_O(\theta)^{\otimes N}$. In the task of localizing point sources, considered here, $\theta$ includes the spatial coordinates $\bar r_k = (\bar{x}_k,\bar{y}_k)$ of each source on the object plane. In this task, the entire fluorescence session time ($\sim 1$ ms) is split into time bins of short duration ($\sim 1$ ns) so that not more than one photon is detected in each time bin. The empty time bins are discarded, while the time bins with detections are numbered with $n=1,...,N$, where $N$ is the total number of photons detected. The point $r_n$ on the image plane where a photon was detected in the $n$th bin is considered as a result of a measurement of the $n$th sample of the field. That is, it is implied that the sources' coordinates do not vary with time and the subsequent acts of spontaneous emission are statistically independent.  Note that the exact times of photon detections may not be resolved by the camera of the microscope and are not necessary for the statistical model. Even if the probability of photon detection in one time bin is as small as $0.01$, the total number of detected photons can be rather high, typically $N\sim10^4$. Specifically, $\rho_O(\theta)$ can be understood as the density operator of the electronic excitation in a time bin where a photon was detected. Alternatively, it can be understood as the density operator of a spontaneously emitted photon on the object plane.

\begin{figure*}[!t]
\centering
\includegraphics[width=\linewidth]{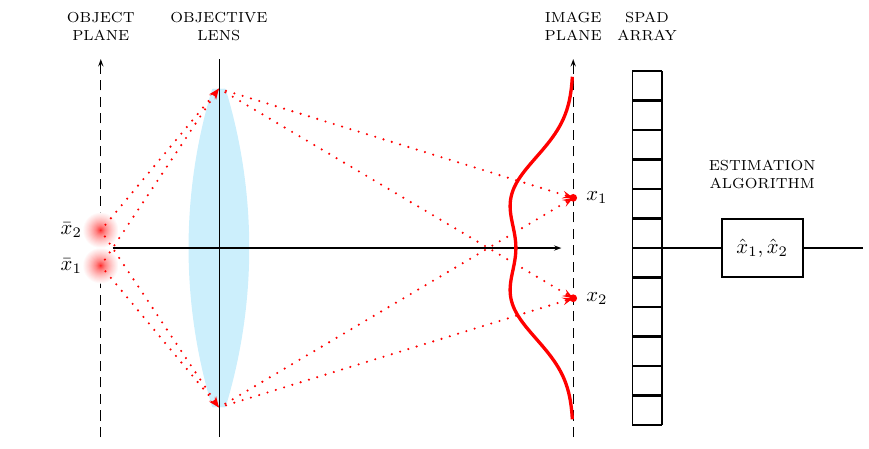}
\caption{Scheme for measuring the positions of two incoherent sources by imaging them in a microscope equipped with a single-photon avalanche diode (SPAD) array and estimating the positions from the measurement data. The dotted red lines show the rays that form the geometrical images at points $x_1$ and $x_2$, while the solid red line shows the field intensity distribution in the image plane as a sum of two point-spread functions centered at $x_1$ and $x_2$. \label{fig:micro}}
\end{figure*}

A generalized quantum measurement consists of measuring a probe system upon its interaction with the object and comprises three main stages: initialization of the probe, interaction between the probe and the object, and measurement of the probe \cite{Barbieri22}. 

In the first stage, the probe is prepared in a state $\dirprod{\rho}_P$, which is, in the simplest case of independent identical systems, just a direct product of the states of each sample of the probe $\rho_P$:  $\dirprod{\rho}_P=\rho_P^{\otimes N}$. However, more complicated entangled probes can be used for reaching the quantum-enhanced sensitivity \cite{Giovannetti06}, which we do not consider here. In optical microscopy, the probe is represented by the part of the optical field captured by the aperture of the microscope, see Fig.~\ref{fig:micro}. For the spontaneous emission of fluorophores, its initial state is vacuum. Thus, its preparation simply means screening the field of view of the microscope from the background light in the frequency band used for imaging. 

In the second stage, the probe and object undergo an interaction, which is described in the most general form by a completely positive trace-preserving linear map $\mathcal{E}$, so that the state of all $N$ samples of the object and probe after the interaction is $\dirprod{\rho}_{OP}(\theta)=\mathcal{E}\left[ \dirprod{\rho}_O(\theta) \otimes\dirprod{\rho}_P\right]$. This map can be made unitary on an extended state space by a Stinespring dilation, that is, by adding up an environment system with the initial state $\dirprod{\rho}_E$ and tracing out the environment after the interaction: $\dirprod{\rho}_{OP}(\theta)=\Tr_E\left\{ \dirprod{U}\dirprod{\rho}_O(\theta) \otimes\dirprod{\rho}_P \otimes\dirprod{\rho}_E\dirprod{U}^\dagger\right\}$, where $\dirprod{U}$ is the unitary operator describing the interaction of the object, probe, and environment and $\dagger$ stands for Hermit conjugation. In the simplest case of independent identical interactions, both the state of the environment and the interaction operator are simple direct products of one-sample operators $\rho_E$ and $U_\theta$:  $\dirprod{\rho}_E=\rho_E^{\otimes N}$ and $\dirprod{U}=U^{\otimes N}$. 
In optical microscopy of fluorophores, the environment is represented by the part of the optical field not captured by the aperture of the microscope and the degrees of freedom of the fluorophore other than the optical transition responsible for the spontaneous emission. The interaction consists in optical excitation of a transition belonging to the environment, from where the excitation is transmitted to the object, which spontaneously emits a photon, belonging to the probe or to the environment.  

In the third stage, the probe is measured giving the outcome $\dirprod{r}$, which is a combination of measurement outcomes $r_n$ for every sample of the probe: $\dirprod{r}=\{r_1,...,r_N\}$. The probability of this outcome is given by the Born rule $\dirprod{p}(\dirprod{r}|\theta) = \Tr_P\left\{\dirprod{\Lambda}(\dirprod{r}) \tilde{\dirprod{\rho}}_P(\theta)\right\}$, where $\tilde{\dirprod{\rho}}_P(\theta) =\Tr_O\left\{\dirprod{\rho}_{OP}(\theta)\right\}$ is the state of the probe with the object traced out, and $\dirprod{\Lambda}(\dirprod{r})$ is the positive-operator-valued measure (POVM) for the outcome $\dirprod{r}$. As before, for independent identical measurements, the POVM is a direct product of POVMs for every sample: $\dirprod{\Lambda}(\dirprod{r})=\otimes_{n=1}^N\Lambda(r_n)$. Collective measurements of probe samples are studied in some scenarios of quantum metrology \cite{Giovannetti06}, but are not considered here. When the magnification of the microscope is sufficiently high, its field of view is covered by a large enough number of pixels, so that the measurement outcome can be represented by the continuous coordinate $r_n=\{x_n,y_n\}$ of the detected photon. The effects of detector pixelation are studied in the literature \cite{Ram06}, but are not considered here. The corresponding single-sample POVM is the projection on the single-photon state at point $r_n$: $\Lambda(r_n)=|r_n\rangle\langle r_n|$, where $|r_n\rangle=a^\dagger(r_n)|\text{vac}\rangle$, $a^\dagger(r_n)$ being the photon creation operator at point $r_n$ and $|\text{vac}\rangle$ is the vacuum state of the field.

Summing up the above relations, we write the conditional probability of the measurement outcome $\dirprod{r}$ under the condition that the parameter is $\theta$ as
\begin{equation}\label{pr}
    \dirprod{p}(\dirprod{r}|\theta) = \Tr\left\{\dirprod{\Lambda}(\dirprod{r}) \dirprod{U}\dirprod{\rho}_O(\theta)\otimes \dirprod{\rho}_P\otimes \dirprod{\rho}_E\dirprod{U}^\dagger\right\},
\end{equation}
where the trace is taken over all $N$ samples of the three systems: object, probe, and environment. The value of the parameter $\theta$ is inferred from the random outcome $\dirprod{r}$ by a function $\hat\theta(\dirprod{r})$ known as the estimator, which is again a random variable. The estimator is called unbiased if its mean coincides with the true value of the parameter,  $\langle\hat\theta(\dirprod{r})\rangle=\theta$, where the angular brackets denote the averaging with the conditional probability, Eq. (\ref{pr}). For any unbiased estimator, its covariance matrix $\Covel(\hat\theta)_{jk}=\langle(\hat\theta_j-\theta_j) (\hat\theta_k-\theta_k)\rangle$ is lower bounded by the Cram\'er-Rao bound discussed in the Introduction for a single parameter and having a matrix form in the multiparameter case \cite{Cox&Hinkley,Lehmann-Casella}:
\begin{equation}\label{CR}
\Cov(\hat\theta) \ge  \dirprod{\mathbf{F}}^{-1}(\theta),
\end{equation}
where the matrix inequality $\mathbf{A}\ge \mathbf{B}$ for two Hermitian matrices $\mathbf{A}$ and $\mathbf{B}$ means that the matrix $\mathbf{A}-\mathbf{B}$ is positive semidefinite, i.e., all its eigenvalues are nonnegative \cite{HornJohnson}, and $\dirprod{\mathbf{F}}(\theta)$ is the Fisher information matrix, whose elements are defined as 
\begin{equation}\label{Fij}
\dirprod{F}_{ij}(\theta) = \int \frac1{\dirprod{p}(\dirprod{r}|\theta)}\frac{\partial \dirprod{p}(\dirprod{r}|\theta)} {\partial\theta_i}\frac{\partial\dirprod{p}(\dirprod{r}|\theta)} {\partial\theta_j} d\dirprod{r}.
\end{equation}
Here and in the following, we assume that $\dirprod{p}(\dirprod{r}|\theta)$ is nonzero everywhere, which is true for an imaging system whose point-spread function is modeled by a Gaussian. Note that we use bold font for matrices, italic for their elements, Roman for the elements of function matrices, and underline, as above, for the quantities related to the entire set of $N$ samples. 

In optical microscopy, the typical estimator for the location of a single emitter is the result of fitting the distribution of the field intensity on its image spot by a two-dimensional Gaussian function \cite{Rust06,Pfender14}. In the case of multiple emitters, the Cram\'er-Rao bound is saturated by finding the parameter $\hat\theta(\dirprod{r})$ which maximizes $\dirprod{p}(\dirprod{r}|\theta)$ considered as a function of $\theta$ and known in this aspect as likelihood, which gives the name of ``maximal likelihood'' to the estimator \cite{Cox&Hinkley,Lehmann-Casella}.

\subsection{Localization efficiency \label{sec:Efficiency}}
The diagonal elements of the matrix $\Cov(\hat\theta)$ give variances of the components of $\theta$, $\Covel(\hat\theta)_{ii}=\Var(\hat\theta_i)$ and traditionally the entire diagonal of the covariance matrix is considered as the indicator of success in determining the unknown parameters. However, when tens or hundreds of sources are available, one might be interested in one number characterizing the efficiency of mass localization. We prefer to express this ``localization efficiency'' through the inverse of the parameter variance known as precision \cite{Bernardo-Smith,Rehacek17}, so that a high efficiency corresponds to a low position variance. This localization efficiency can be introduced in several alternative ways, depending on the order in which the averaging and inversion are performed.

\newtheorem{definition}{Definition}
\begin{definition}\label{definition:Vbar}
For any unbiased estimator $\hat\theta$ of an $M$-dimensional parameter $\theta$, the average total precision is 
\begin{equation}\label{Htot}
H_\textnormal{tot}(\hat\theta) = \frac{M}{\Tr\{\Cov(\hat\theta)\}}.
\end{equation}
\end{definition}
This quantity is the inverse of the average variance and, according to Eq. (\ref{CR}), satisfies the inequality
\begin{equation}
H_\textnormal{tot}(\hat\theta) \le \frac{M}{\Tr\{\dirprod{\mathbf{F}}^{-1}(\theta)\}},   
\end{equation}
which is asymptotically saturated by the maximum-likelihood estimator. 

However, the average variance can be odd suited to the description of the case where a localization of multiple fluorophores is required. Imagine a situation where one of the 100 fluorophores does not emit light for some reason. In this case, the probability distribution does not depend on the position of this fluorophore, and therefore all rows and columns of the Fisher information matrix corresponding to its coordinates are zeros. The variance of these coordinates is infinite and the average variance is infinite as well, giving $H_\textnormal{tot}(\hat\theta)=0$, which means failure of the estimation procedure. On the other hand, absence of one of 100 fluorophores is not important for the task of visualization of a biological object, for which the fluorophores are intended. Thus, we arrive at another measure of the localization efficiency where the individual precisions are averaged rather than the individual variances. 
\begin{definition}\label{definition:Hind}
For any unbiased estimator $\hat\theta$ of an $M$-dimensional parameter $\theta$, the average individual precision is 
\begin{equation}\label{Hind}
H_\textnormal{ind}(\hat\theta) = \frac1M\sum_{i=1}^M\frac1{\Covel(\hat\theta)_{ii}}.
\end{equation}
\end{definition}

In the considered case of one dark fluorophore, its individual precision $1/\Covel(\hat\theta)_{ii}$ is zero and does not contribute to the average individual precision defined by Eq. (\ref{Hind}). The upper bound for the average individual precision is set according to Eq. (\ref{CR}), by the inequality $H_\text{ind}(\hat\theta) \le \bar F_\text{ind}(\theta)$, where 
\begin{equation}\label{Fbarind}
\bar F_\text{ind}(\theta) = \frac1M\sum_{i=1}^M\frac1{\left[\dirprod{\mathbf{F}}^{-1} (\theta)\right]_{ii}}
\end{equation}
is the average individual Fisher information. This bound is asymptotically saturated by the maximum-likelihood estimator. 

Both $H_\text{ind}(\hat\theta)$ and $\bar F_\text{ind}(\theta)$ share the drawback of being non-invariant with respect to rotations in the parameter space. Aiming at an invariant measure, we introduce one more measure of localization efficiency. First, we define the precision matrix \cite{Bernardo-Smith} $\Prec(\hat\theta) = \Cov(\hat\theta)^{-1}$, whose diagonal entry $\Precel(\hat\theta)_{ii}$ is larger than or equal to the individual precision $1/\Covel(\hat\theta)_{ii}$. If we consider an orthogonal rotation in the parameter space, $\varphi = \mathbf{O}^T\theta$, where $\varphi=\{\varphi_1,...,\varphi_M\}$, both parameter vectors are understood as column vectors, $T$ denotes a transposition, and $\mathbf{O}$ is an orthogonal matrix such that $\Cov(\hat\varphi)=\mathbf{O}^T\Cov(\hat\theta)\mathbf{O}$ is diagonal with $\hat\varphi=\mathbf{O}^T\hat\theta$ an unbiased estimator for $\varphi$, then $\Precel(\hat\varphi)_{ii}=1/\Covel(\hat\varphi)_{ii}$. We will call $\varphi$ the eigenparameter of the localization task. Now, we define the third measure of localization efficiency as follows.
\begin{definition}\label{definition:Heig}
For any unbiased estimator $\hat\theta$ of an $M$-dimensional parameter $\theta$, the average eigenparameter precision is 
\begin{equation}\label{Heig}
H_\textnormal{eig}(\hat\theta) = \frac1M\Tr\{\Prec(\hat\theta)\}.
\end{equation}
\end{definition}

The upper bound for this quantity can be found from Eq. (\ref{CR}). For any two invertible Hermitian matrices $\mathbf{A}$ and $\mathbf{B}$, $\mathbf{A}\ge \mathbf{B}$ if and only if $\mathbf{A}^{-1}\le \mathbf{B}^{-1}$ \cite{HornJohnson}. Note that the covariance, precision and Fisher information matrices are real symmetric by definition, and as such are Hermitian. Therefore, Eq. (\ref{CR}) can be rewritten as $\Prec(\hat\theta) \le \dirprod{\mathbf{F}}(\theta)$. Taking a trace of both sides of this inequality, we obtain $H_\text{eig}(\hat\theta) \le \bar F(\theta)$, where 
\begin{equation}
\bar F(\theta) = \frac1M\Tr\left\{\dirprod{\mathbf{F}}(\theta)\right\}   
\end{equation}
is the average Fisher information. This bound is asymptotically saturated by the maximum-likelihood estimator. The matrix $\mathbf{O}$ diagonalizing the covariance matrix coincides asymptotically with that diagonalizing the Fisher information matrix. However, the matrix $\mathbf{O}$ is itself dependent on the unknown parameters and is typically unknown when the localization procedure starts. 

Two commentaries on the introduced measures are necessary. It is typical for multiparameter quantum metrology to introduce a weight matrix assigning different statistical weights to different parameters and to consider the trace of the product of the covariance matrix with the weight matrix as a measure of the multiparameter estimation success \cite{Albarelli20,Barbieri22}. We do not see why, in the task of multiple source localization, one source might have preference before another and therefore assign equal weights to the coordinates of all sources, which leads us to a simple summation in all three definitions.

It should also be understood that the measures given by Definitions \ref{definition:Hind} and \ref{definition:Heig} are not intended for use in the fitting algorithms. The quantities $H_\text{ind}(\hat\theta)$ and $H_\text{eig}(\hat\theta)$ may be very large in the case where only one source is localized with very high precision but the positions of all other sources are totally unknown, which is not the best localization outcome. We assume that the fitting is realized by the maximum likelihood estimation and the measure of localization efficiency is applied afterwards to the results of such estimation. The evaluation of efficiency here is similar to that in quantum teleportation \cite{Braunstein98,Horoshko00}: the verifier knowing the exact positions of the sources gives a sample to the microscopist, who returns the estimated source positions, for which the verifier calculates the measure of efficiency, similar in this respect to the teleportation fidelity. 

The three introduced measures are ordered as follows:
\begin{equation}\label{ordering}
H_\text{tot}(\hat\theta)\le H_\text{ind}(\hat\theta)\le H_\text{eig}(\hat\theta).
\end{equation}
The first inequality follows from the convexity of the function $f(x)=1/x$ on the interval $x\in(0,+\infty)$, resulting in 
\begin{equation}
f\left(\frac1M\sum\limits_{i=1}^Mx_i\right)\le \frac1M\sum\limits_{i=1}^M f(x_i).
\end{equation}
Substituting $x_i=\Cov(\hat\theta)_{ii}$, we obtain the first inequality. 
To prove the second inequality, consider an invertible real symmetric $M\times M$ matrix $\mathbf{A}$ and a unit $M\times 1$ vector $v$. Let $\mathbf{O}$ be an orthogonal matrix that brings $\mathbf{A}$ to the diagonal form, $\mathbf{A}=\mathbf{O}\mathbf{D}\mathbf{O}^T$, where $\mathbf{D}$ is diagonal. Then $v^T\mathbf{A}v=\tilde{v}^T\mathbf{D}\tilde{v}$, where $\tilde{v}=\mathbf{O}^Tv$. Since $\tilde{v}$ is also a unit vector, the numbers $\tilde{v}_i^2$ sum up to unity and can be considered as statistical weights. From the convexity of the function $f(x)=1/x$ on the interval $x\in(0,+\infty)$, as above, we obtain 
\begin{equation}
f\left(\sum\limits_{i=1}^M\tilde{v}_i^2D_{ii}\right)\le \sum\limits_{i=1}^M \tilde{v}_i^2 f(D_{ii}),
\end{equation}
or $(\tilde{v}^T\mathbf{D}\tilde{v})^{-1}\le \tilde{v}^T\mathbf{D}^{-1}\tilde{v}$ and, as a consequence, $(v^T\mathbf{A}v)^{-1}\le v^T\mathbf{A}^{-1}v$. Choosing a vector $v_i=\delta_{im}$ for some $1\le m\le M$, we obtain the inequality $1/A_{mm}\le(\mathbf{A}^{-1})_{mm}$. Substituting $A_{mm}=\Covel(\hat\theta)_{mm}$ and summing up over $m$, we obtain the second inequality of Eq.~(\ref{ordering}).

\subsection{Quantum Cram\'er-Rao bound \label{sec:QFI}}
For a given state of the object, the conditional probabilities, Eq. (\ref{pr}), and therefore the parameter variances calculated from them, depend on four quantities: $\dirprod{\rho}_P$, $\dirprod{\rho}_E$, $\dirprod{U}_\theta$, and $\dirprod{\Lambda}(\dirprod{r})$. Each of these quantities can be manipulated to optimize the estimation. In particular, optimization of $\dirprod{\Lambda}(\dirprod{r})$ corresponds to the choice of the best measurement for the given state of the probe $\tilde{\dirprod{\rho}}_P(\theta)$. For this purpose, quantum Fisher information matrix $\dirprod{\mathbf{Q}}(\theta)$ is defined as \cite{Helstrom-book,HolevoBook82,Paris09}
\begin{equation}\label{Qij}
\dirprod{Q}_{ij}(\theta) = \frac12\Tr\left\{
\left(\mathcal{L}_i\mathcal{L}_j+\mathcal{L}_j\mathcal{L}_i\right) \tilde{\dirprod{\rho}}_P(\theta)
\right\},   
\end{equation}
where $\mathcal{L}_i$ is the symmetric logarithmic derivative operator defined by the equation
\begin{equation}\label{SLD}
\frac{\partial}{\partial\theta_i} \tilde{\dirprod{\rho}}_P(\theta) = \frac12\left[\mathcal{L}_i\tilde{\dirprod{\rho}}_P(\theta) + \tilde{\dirprod{\rho}}_P(\theta)\mathcal{L}_i\right]. 
\end{equation}

Quantum Fisher information matrix satisfies the condition $\dirprod{\mathbf{F}}^{-1}(\theta)\ge \dirprod{\mathbf{Q}}^{-1}(\theta)$, where $\dirprod{\mathbf{F}}^{-1}(\theta)$ is calculated for the same state of the probe with an arbitrary POVM $\dirprod{\Lambda}(\dirprod{r})$. Moreover, there is often an optimal POVM $\dirprod{\Lambda}_\text{opt}(\dirprod{r})$ (but not always \cite{Albarelli20}) for which $\dirprod{\mathbf{F}}_\text{opt}^{-1}(\theta) = \dirprod{\mathbf{Q}}^{-1}(\theta)$. This means that the quantity $\Tr\{\dirprod{\mathbf{Q}}^{-1}(\theta)\}/M$ gives the lower bound of the average variance for the optimized measurement. As mentioned above, manipulation of $\dirprod{\rho}_P$ by entangling the probes is known to further minimize the estimation error in some tasks of quantum metrology \cite{Giovannetti06}. In this paper, we show that the superresolution obtained by making the sources blink corresponds to manipulation of  $\dirprod{\rho}_O(\theta)$ without affecting the value of the parameter $\theta$.

\section{Additivity and weak convexity of Fisher information matrix \label{sec:Additivity}}

Fisher information matrix possesses some remarkable properties, which are summarized below. First, we reformulate its standard definition \cite{Cox&Hinkley,Lehmann-Casella} in terms of a functional of the probability density.

\begin{definition}\label{definition:F}
Given two positive integers $\ell$ and $M$ and a random variable $X\in\mathbb{R}^\ell$ depending on parameter $\theta\in\mathbb{R}^M$, so that the probability density for $X$ under condition that the parameter is $\theta$ is $q(X|\theta)$, the Fisher information matrix for the information contained in $X$ about $\theta$ is the $M\times M$-matrix-valued functional $\bm{\mathcal{F}}_M^{(\ell)}\left[q(X|\theta)\right]$, whose matrix elements are
\begin{equation}\label{Flm}
\mathcal{F}_M^{(\ell)}\left[q(X|\theta)\right]_{ij} = \int_{\mathbb{R}^\ell} \frac1{q(X|\theta)}\frac{\partial q(X|\theta)}{\partial\theta_i}\frac{\partial q(X|\theta)}{\partial\theta_j} d^\ell X.
\end{equation}
\end{definition}

Now, we formulate the property of additivity for this functional.

\newtheorem{lemma}{Lemma}
\begin{lemma}\label{lemma:additivity}
The Fisher information matrix possesses the property of
additivity: 
    \begin{eqnarray}
&&\bm{\mathcal{F}}_M^{(\ell_1+\ell_2)}\left[q_1(X_1|\theta)q_2(X_2|\theta)\right]\\\nonumber
&&=\bm{\mathcal{F}}_M^{(\ell_1)}\left[q_1(X_1|\theta)\right] +\bm{\mathcal{F}}_M^{(\ell_2)}\left[q_2(X_2|\theta)\right],
\end{eqnarray}
where $X_1\in\mathbb{R}^{\ell_1}$ and $X_2\in\mathbb{R}^{\ell_2}$.
\end{lemma}

The proof can be found in the literature \cite{Lehmann-Casella}. Next, we formulate the property of convexity in the single-parameter case proven by Cohen~\cite{Cohen68}.

\begin{lemma}[Cohen]\label{lemma:convexity}
In the single-parameter case, $M=1$, the Fisher information matrix (reducing to a scalar in this case) possesses the property of strict convexity:
\begin{eqnarray}
     &&\mathcal{F}_1^{(\ell)}\left[\gamma q_1(X|\theta)+(1-\gamma)q_2(X|\theta)\right] \\\nonumber
     &&< \gamma\mathcal{F}_1^{(\ell)}\left[q_1(X|\theta)\right] +(1-\gamma)\mathcal{F}_1^{(\ell)}\left[q_2(X|\theta)\right], 
\end{eqnarray}
for any $0<\gamma<1$ and $q_1(X|\theta)\not\equiv q_2(X|\theta)$.
\end{lemma}

This property can be extended to the multiple-parameter case by the following theorem.

\newtheorem{theorem}{Theorem}
\begin{theorem}\label{theorem:convexity}
Given two positive integers $K\ge2$ and $d$, in the general case of $M$ parameters split into $K$ groups of $d$ parameters each, $\theta=\{\bar\theta_1,...,\bar\theta_K\}$, where $\bar\theta_k\in\mathbb{R}^d$, $M=Kd$, and given real numbers $\mu_1$,...,$\mu_K$ belonging to $[0,1)$ and summing up to unity, the Fisher information matrix possesses the following properties. 
\begin{enumerate}
\item[A.] Strict weak matrix convexity in the case of equal weights $\mu_k=1/K$:
\begin{equation}\label{convexity}    \bm{\mathcal{F}}_M^{(\ell)}\left[\frac1K\sum\limits_{k=1}^{K} q_k(X|\theta)\right] < \frac1K\sum\limits_{k=1}^{K} \bm{\mathcal{F}}_M^{(\ell)}\left[q_k(X|\theta)\right],
\end{equation}
under condition that the probability density $q_k(X|\theta)$ depends on $\bar\theta_k$ only and 
\begin{equation}\label{condition}
    \mathcal{F}_M^{(\ell)} \left[q_k(X|\theta)\right]_{ij} = \left\{ \begin{array}{ll}
       f_0,  &  i=j\in[(k-1)d+1,kd],\\
       0,  & \mathrm{otherwise},
    \end{array}\right.
\end{equation}
where $f_0$ is a positive number independent of $k$.
\item[B.] Strict trace convexity in the general case of arbitrary weights $\mu_k$:
\begin{equation}\label{trace-convexity}    
\Tr\left\{\bm{\mathcal{F}}_M^{(\ell)} \left[\sum\limits_{k=1}^{K} \mu_k q_k(X|\theta)\right]\right\} < \sum\limits_{k=1}^{K}\mu_k \Tr\left\{\bm{\mathcal{F}}_M^{(\ell)}\left[q_k(X|\theta) \right]\right\}.
\end{equation}
\end{enumerate}
\end{theorem}
\begin{proof}
Let us denote the matrix in the left-hand side of Eq.~(\ref{trace-convexity}) by $\mathbf{F}$ and that in the right-hand side by $\mathring{\mathbf{F}} = \sum_k\mu_k \bm{\mathcal{F}}_M^{(\ell)}\left[q_k(X|\theta)\right]$. $\mathbf{F}$ is a real symmetric matrix and as such can be brought to a diagonal form by an orthogonal rotation in the parameter space: $\mathbf{F} =\mathbf{O}\mathbf{D}_F\mathbf{O}^T$, where $\mathbf{D}_F$ is diagonal and $\mathbf{O}$ is orthogonal. Consider a linear transformation from the (column) vector of parameters $\theta=\{\theta_1,...,\theta_M\}$ to the new vector of parameters $\varphi=\{\varphi_1,...,\varphi_M\}$, $\varphi=\mathbf{O}^T\theta$. The conditional probability density of $X$ for given $\varphi$ is $\tilde q(X|\varphi)=q(X|\mathbf{O}\varphi)$. Let us denote the Fisher information matrix for the information contained in $X$ about $\varphi$ by $\tilde{\mathbf{F}}$. According to Definition~\ref{definition:F}, its elements are 
\begin{eqnarray}\nonumber
\mathcal{F}_M^{(\ell)}\left[\tilde q(X|\varphi)\right]_{ij} &=& \int_{\mathbb{R}^\ell} \frac1{\tilde q(X|\varphi)}\frac{\partial \tilde q(X|\varphi)}{\partial\varphi_i}\frac{\partial\tilde q(X|\varphi)}{\partial\varphi_j} d^\ell X \\
&=& \sum_{kl}  \mathcal{F}_M^{(\ell)}\left[q(X|\theta)\right]_{kl} \frac{\partial\theta_k}{\partial\varphi_i} \frac{\partial\theta_l}{\partial\varphi_j}\\\nonumber
&=& \sum_{kl}  \mathcal{F}_M^{(\ell)}\left[q(X|\theta)\right]_{kl} O_{ki} O_{lj},
\end{eqnarray}
or $\tilde{\mathbf{F}} = \mathbf{O}^T\mathbf{F}\mathbf{O} = \mathbf{D}_F$. This means that $\varphi$ can be considered as a vector of ``eigenparameters'' of the matrix $\mathbf{F}$. Let us fix all eigenparameters except for one, $\varphi_i$, and write the conditional probability density of $X$ for given $\varphi_i$ as $\tilde q'(X|\varphi_i)=\tilde q(X|\varphi)$. The Fisher information contained in $X$ about $\varphi_i$ is
\begin{eqnarray}\label{rotation}
\mathcal{F}_1^{(\ell)}\left[\tilde q'(X|\varphi_i)\right] &=& \mathcal{F}_M^{(\ell)}\left[\tilde q(X|\varphi)\right]_{ii} \\\nonumber
&=& \sum_{mn}\mathcal{F}_M^{(\ell)} \left[q(X|\theta)\right]_{mn} O_{mi} O_{ni}
\end{eqnarray}
According to Lemma~\ref{lemma:convexity}, it satisfies the strict convexity inequality 
\begin{equation}\label{ineq1}
\mathcal{F}_1^{(\ell)}\left[
\sum\limits_{k=1}^{K} \mu_k\tilde q_k'(X|\varphi_i)\right]
< 
\sum\limits_{k=1}^{K}\mu_k\mathcal{F}_1^{(\ell)}\left[\tilde q_k'(X|\varphi_i)\right]. 
\end{equation} 
Applying Eq.~(\ref{rotation}) to each side of this inequality, we obtain 
\begin{eqnarray}\label{ineq2}
&&\sum_{mn}\mathcal{F}_M^{(\ell)} \left[
\sum\limits_{k=1}^{K} \mu_k q_k(X|\theta)
\right]_{mn} O_{mi} O_{ni}\\\nonumber
&&< 
\sum_{mn}O_{mi} O_{ni}\sum\limits_{k=1}^{K}
\mu_k\mathcal{F}_M^{(\ell)}\left[ q_k(X|\theta)\right]_{mn}, 
\end{eqnarray} 
or, in matrix form, as
\begin{equation}\label{ineq3}
\left\{\mathbf{O}^T\left(\mathring{\mathbf{F}}-\mathbf{F}\right) \mathbf{O}\right\}_{ii}>0. 
\end{equation} 

In the case of part A, we have $\mathring{\mathbf{F}} = f_0\mathbf{I}_M/K$, where $\mathbf{I}_M$ is the unit $M\times M$ matrix, which can be seen applying the condition of equal weights, $\mu_k=1/K$, and then the identity $\sum_{k}\mathcal{F}_M^{(\ell)}\left[ q_k(X|\theta)\right]_{mn} = f_0\delta_{mn}$ following from Eq.~(\ref{condition}). Since inequality (\ref{ineq3}) holds for any $i$ and both matrices $\mathbf{O}^T\mathring{\mathbf{F}}\mathbf{O}=\mathring{\mathbf{F}}$ and $\mathbf{O}^T\mathbf{F}\mathbf{O}=\mathbf{D}_F$ are diagonal, this inequality can be rewritten as $\mathbf{O}^T\left(\mathring{\mathbf{F}}-\mathbf{F}\right) \mathbf{O}>0$ or $\mathbf{F}<\mathring{\mathbf{F}}$, which gives Eq.~(\ref{convexity}). Part A of the theorem is proved.

In the general case, we sum Eq. (\ref{ineq3}) over $i$ and use the invariance of the trace with respect to orthogonal rotations to obtain $\Tr \mathbf{F} < \Tr \mathring{\mathbf{F}}$, which coincides with Eq. (\ref{trace-convexity}) and proves part B.
\end{proof}

Note that the area of applicability of the matrix convexity is slightly wider than that stated in part A of Theorem \ref{theorem:convexity}. Indeed, what we need for the proof, is that $\mathring{\mathbf{F}}$ is proportional to the unit matrix, which can be reached by replacing $f_0$ by $f_k$ in Eq. (\ref{condition}) and demanding, instead of equal weights, that $\mu_kf_k$ are independent of $k$.

\section{Localization of blinking sources \label{sec:Localization}}

Now, we consider $K$ point-like mutually incoherent sources of light located at unknown locations $\{\bar{x}_k,\bar{y}_k\}$ on the object plane of the microscope. The sources have brightnesses $B_k$, summing up to the full brightness $B=\sum_kB_k$. We define relative brightnesses $\mu_k=B_k/B$ determining the part of a given source in the total photon flux. We consider two scenarios for estimating the $K$ unknown locations, see Fig.~\ref{fig:counts}.

\begin{figure}[!ht]
\centering
\includegraphics[width=\linewidth]{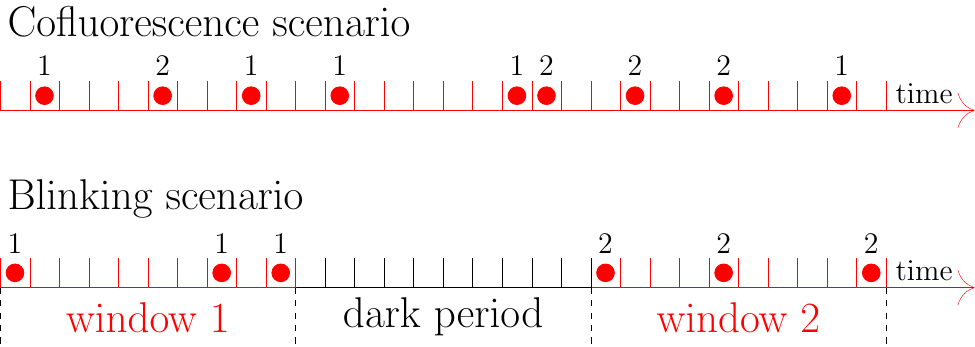}
\caption{Photocounts (red balls) appearing in certain time bins in two scenarios of fluorescence. The number above a ball shows from what source the photon comes. \label{fig:counts}}
\end{figure}

In the first scenario, which we call ''cofluorescence'', all sources can emit light simultaneously. Every detected photon can come from the $k$th source with the probability $\mu_k$, so that the state of the source in each time bin is $\rho_O^\text{coflu}=\sum_k\mu_k\rho_{Ok}$, where $\rho_{Ok}$ is the state of the object with just the $k$th source active, and the state of the object samples in all time bins where the photons are detected is 
\begin{equation}\label{rhomix}
\dirprod{\rho}_O^\text{coflu}=\left(\sum_k\mu_k\rho_{Ok}\right)^{\otimes N}. 
\end{equation}
All other quantities in the right-hand side of Eq. (\ref{pr}) correspond to independent identical generalized measurements, that is they are simple $N$th powers of the single-sample quantities: 
\begin{equation}\label{iigm}
\dirprod{\rho}_P=\rho_P^{\otimes N},\,\, \dirprod{\rho}_E=\rho_E^{\otimes N},\,\, \dirprod{U}=U^{\otimes N},\,\, \dirprod{\Lambda}(\dirprod{r})=\prod\limits_n\Lambda(r_n).
\end{equation}
Substituting these expressions into Eq. (\ref{pr}), we obtain $\dirprod{p}(\dirprod{r}|\theta)=\prod_np(r_n|\theta)$, where 
\begin{equation}\label{prsample}
    p(r|\theta) = \Tr\left\{\Lambda(r)U\rho_O^\text{coflu}(\theta)\otimes\rho_P\otimes\rho_EU^\dagger\right\}
\end{equation}
is the single-sample conditional probability of measurement outcome $r=\{x,y\}$, which is the location of the detected photon. Substituting this probability into Eq. (\ref{Fij}), we find the Fisher information matrix in the form $\dirprod{\mathbf{F}}_\text{coflu}(\theta)=N\mathbf{F}_\text{coflu}(\theta)$, where $\mathbf{F}_\text{coflu}(\theta)$ is the single-sample Fisher information matrix with the elements
\begin{equation}\label{Fijsample}
F_{\text{coflu},ij}(\theta) = \int \frac1{p(r|\theta)}\frac{\partial p(r|\theta)}{\partial\theta_i}\frac{\partial p(r|\theta)}{\partial\theta_j} dr.
\end{equation}

In the second scenario, which we call ''blinking'', the object is affected so that just one source is active in a given time window separated from the neighboring windows by dark periods. Let us denote by $T_k^{(i)}$ the duration of the window of the $i$th activity of the $k$th source. This duration is assumed to be long enough for a successful identification of the source number by the localization algorithm. We unite all windows of the same source into the $k$th super-window of duration $T_k=\sum_iT_k^{(i)}$. For a sufficiently long experiment, the durations of activity of the sources are approximately equal, $T_k=T$. Therefore, each source emits the total number of $N_k=\mu_kN$ photons (we assume that $\mu_kN$ is an integer, which is a good approximation for a high $N$). The state of all $N_k$ samples in the $k$th super-window is $\rho_{Ok}^{\otimes N_k}$. The joint state of all $N$ samples of the source is
\begin{equation}\label{rhoblink}
\dirprod{\rho}_O^\text{blink} = \rho_{O1}^{\otimes N_1}\otimes...\otimes\rho_{OK}^{\otimes N_K}.
\end{equation}
All other quantities in the right-hand side of Eq. (\ref{pr}) are simple $N$th powers of the single-sample quantities, as in the previous scenario. Substituting these expressions into Eq. (\ref{pr}), we obtain $\dirprod{p}(\dirprod{r}|\theta)=\prod_{k=1}^K \dirprod{p}_k(\dirprod{r}_k|\theta)$, where $\dirprod{r}_k$ is a part of $\dirprod{r}$ with the outcomes in the $k$th super-window and $\dirprod{p}_k(\dirprod{r}_k|\theta) =\prod_{l=1}^{N_k}p_k(r_{kl}|\theta)$ with
\begin{equation}\label{prksample}
p_k(r|\theta) = \Tr\left\{\Lambda(r)U\rho_{Ok}(\theta)\otimes\rho_P\otimes\rho_EU^\dagger\right\}
\end{equation}
the single-sample conditional probability of measurement outcome when the $k$th source is only active. Here, the index $kl$ means the $l$th element in the $k$th super-window and equals $kl=l+\sum_{m=1}^{k-1}N_m$. Substituting this probability into Eq. (\ref{Fij}), we find the Fisher information matrix in the form $\dirprod{\mathbf{F}}_\text{blink}(\theta)=\sum_kN_k\mathbf{F}^{[k]}(\theta)$, where $\mathbf{F}^{[k]}(\theta)$ is the single-sample Fisher information matrix in the $k$th window, where just the $k$th source is active:
\begin{equation}\label{Fijksample}
F_{ij}^{[k]}(\theta) = \int \frac1{p_k(r|\theta)}\frac{\partial p_k(r|\theta)}{\partial\theta_i}\frac{\partial p_k(r|\theta)}{\partial\theta_j} dr.
\end{equation}
Actually, $\mathbf{F}^{[k]}(\theta)$ depends on the location $\bar\theta_k=\bar r_k$ of the $k$th source only, because other sources are inactive in this super-window.

In order to satisfy the condition of Theorem \ref{theorem:convexity}, we need to restrict our consideration to microscopy with certain properties, typical of ordinary microscopes withing their field of view \cite{Goodman-book}.

\begin{definition}\label{definition:micro}
Microscopy with magnification $g\in\mathbb{R}$ is translation and rotation invariant if the conditional distribution of the coordinates $r=\{x,y\}$ of a detected photon on the image plane under condition of a single source at point $\bar{r}=\{\bar{x},\bar{y}\}$ on the object plane, $p(r|\bar{r})$, satisfies the conditions
\begin{eqnarray}\label{translation}
p(r+g\bar{r}_0|\bar{r}+\bar{r}_0) &=& p(r|\bar{r}),\\\label{rotation2}
p(\mathbf{R}r|\mathbf{R}\bar{r}) &=& p(r|\bar{r}).
\end{eqnarray}
where $\bar{r}_0=\{\bar{x}_0,\bar{y}_0\}$ is an arbitrary real vector and $\mathbf{R}$ is a rotation matrix on the object plane.
\end{definition}

Such a microscopy is characterized by a special form of the Fisher information matrix.

\begin{lemma}\label{lemma:F2}
In translation and rotation invariant microscopy, the Fisher information matrix on the position of a single source  $\bar{r}=\{\bar{x},\bar{y}\}$ is proportional to the unit $2\times2$ matrix $\mathbf{I}$,
\begin{equation}\label{invariance} \bm{\mathcal{F}}_2^{(2)}\left[p(r|\bar{r})\right] = f_0 \mathbf{I},
\end{equation}
where $f_0$ is a positive number independent of $\bar{r}$.
\end{lemma}
\begin{proof}
Rotating the parameters as $\bar{r}\to \mathbf{R}\bar{r}$, we find from Eq.~(\ref{Flm}) the corresponding rotation of the Fisher information matrix, 
$\bm{\mathcal{F}}_2^{(2)}\left[p(r|\mathbf{R}\bar{r})\right] =\mathbf{R}\bm{\mathcal{F}}_2^{(2)}\left[p(r|\bar{r})\right]\mathbf{R}^T$. Since the Jacobian corresponding to a rotation of integration coordinates is $|\det \mathbf{R}|=1$, we also have $\bm{\mathcal{F}}_2^{(2)}\left[p(\mathbf{R}r|\mathbf{R}\bar{r})\right] =\bm{\mathcal{F}}_2^{(2)}\left[p(r|\mathbf{R}\bar{r})\right]$. Combining these two equalities with Eq.~(\ref{rotation2}), we obtain invariance of $\bm{\mathcal{F}}_2^{(2)}\left[p(r|\bar{r})\right]$ to rotations, which means that it is proportional to the unit matrix $\mathbf{I}$. Substituting Eq.~(\ref{translation}) into Eq.~(\ref{Flm}), we find $\bm{\mathcal{F}}_2^{(2)}\left[p(r|\bar{r}+\bar{r}_0)\right] =\bm{\mathcal{F}}_2^{(2)}\left[p(r|\bar{r})\right]$, which means that $f_0$ is independent of $\bar{r}$.
\end{proof}

Now, we formulate the following theorem.

\begin{theorem}\label{theorem:Main}
For $\dirprod{\mathbf{F}}_\text{coflu}(\theta)$ defined with the object state (\ref{rhomix}), $\dirprod{\mathbf{F}}_\text{blink}(\theta)$ defined with the object state (\ref{rhoblink}), and all generalized measurements being independent and identical, the following inequalities hold for any $\rho_{Ok}$, $\rho_P$, $\rho_E$, $U$, and $\Lambda(r)$.
\begin{enumerate}
    \item[A.] In the case of translation and rotation invariant microscopy with $N_k=N/K$,
    \begin{equation}\label{Main}
    \dirprod{\mathbf{F}}_\text{blink}(\theta) > \dirprod{\mathbf{F}}_\text{coflu}(\theta).
    \end{equation}
    \item[B.] In the general case of arbitrary $N_k$, at least two of which are non-zero,
    \begin{equation}\label{MainTrace}
    \Tr\{\dirprod{\mathbf{F}}_\text{blink}(\theta)\} > \Tr\{\dirprod{\mathbf{F}}_\text{coflu}(\theta)\}.
    \end{equation}
\end{enumerate}

\end{theorem}
\begin{proof}
Substituting Eqs. (\ref{rhomix}) and (\ref{iigm}) into Eq. (\ref{pr}) and using the property $\Tr\{\mathbf{A}\otimes \mathbf{B}\}=\Tr \mathbf{A}\Tr \mathbf{B}$, we obtain $\dirprod{p}(\dirprod{r}|\theta)=\prod_np(r_n|\theta)$ with $p(r|\theta)$ defined by Eq. (\ref{prsample}), which can be rewritten as 
\begin{equation}\label{prcoflu}
    p(r|\theta) = \sum_k \mu_k p_k(r|\theta),
\end{equation}
where $p_k(r|\theta)$ is defined by Eq. (\ref{prksample}). Substituting this form of $\dirprod{p}(\dirprod{r}|\theta)$ into Eq. (\ref{Fij}) and using the additivity of Fisher information by Lemma~\ref{lemma:additivity}, we obtain 
\begin{equation}\label{Fcoflu}
\dirprod{\mathbf{F}}_\text{coflu}(\theta) = N \mathbf{F}_\text{coflu}(\theta),
\end{equation}
where $\mathbf{F}_\text{coflu}(\theta)$ is a single-sample Fisher information matrix calculated as a functional of $p(r|\theta)$, $\mathbf{F}_\text{coflu}(\theta)=\bm{\mathcal{F}}_M^{(2)}[p(r|\theta)]$, where the upper index $\ell=2$ is the dimension of the vector $r=\{x,y\}$.

On the other hand, substituting Eqs. (\ref{rhoblink}) and (\ref{iigm}) into Eq. (\ref{pr}), we obtain $\dirprod{p}(\dirprod{r}|\theta) =\prod_{k=1}^K\dirprod{p}_k(\dirprod{r}_k|\theta)$, where $\dirprod{p}_k(\dirprod{r}_k|\theta) =\prod_{l=1}^{N_k}p_k(r_{kl}|\theta)$ and $p_k(r_{kl}|\theta)$ is given by Eq. (\ref{prksample}). Substituting this form of $\dirprod{p}(\dirprod{r}|\theta)$ into Eq. (\ref{Fij}) and using the additivity of Fisher information by Lemma~\ref{lemma:additivity}, we obtain 
\begin{equation}\label{Fblink}
\dirprod{\mathbf{F}}_\text{blink}(\theta) = \sum_k N_k \mathbf{F}^{[k]}(\theta),
\end{equation}
where $\mathbf{F}^{[k]}(\theta)$ is given by Eq. (\ref{Fijksample}) and calculated as a functional of $p_k(r|\theta)$, $\mathbf{F}^{[k]}(\theta)=\bm{\mathcal{F}}_M^{(2)}[p_k(r|\theta)]$. Finally, we apply the equality $N_k=N\mu_k$. In the case of $N_k=N/K$, using the convexity of the Fisher information matrix by part A of Theorem \ref{theorem:convexity}, we obtain Eq. (\ref{Main}).  The condition (\ref{condition}) is satisfied according to Lemma~\ref{lemma:F2}. In the case of arbitrary $N_k$, using the convexity of the trace of the Fisher information matrix by part B of Theorem \ref{theorem:convexity}, we obtain Eq. (\ref{MainTrace}).
\end{proof}

If the matrices $\mathbf{A}$ and $\mathbf{B}$ are invertible and positive, $\mathbf{A} > \mathbf{B}$ is equivalent to $\mathbf{A}^{-1} < \mathbf{B}^{-1}$ \cite{HornJohnson}. As a corollary to Theorem~\ref{theorem:Main}, we have $\dirprod{\mathbf{F}}_\text{blink}^{-1}(\theta) < \dirprod{\mathbf{F}}_\text{coflu}^{-1}(\theta)$, and therefore, for a saturated Cram\'er-Rao bound, in the case of equal brightnesses,
\begin{equation}\label{advantage}
H_\text{blink}(\hat\theta) > H_\text{coflu}(\hat\theta),
\end{equation}
where $H(\hat\theta)$ may be any of the three measures of localization efficiency introduced in Sec.~\ref{sec:Efficiency}, $H_\text{tot}(\hat\theta)$, $H_\text{ind}(\hat\theta)$, or $H_\text{eig}(\hat\theta)$. Moreover, Eq.~(\ref{advantage}) holds for $H_\text{eig}(\hat\theta)$ also in the case of arbitrary brightnesses.

\section{Localization of two incoherent sources \label{sec:Example}}

In this section, we consider a simple example illustrating the theory developed above and consisting of a one-dimensional localization of two incoherent sources lying at unknown locations $\bar{x}_1$ and $\bar{x}_2$ on the object plane of the microscope and having known relative brightnesses $\mu_1$ and $\mu_2$. On its image plane, they are represented by two spots centered at $x_1=g\bar x_1$ and $x_2=g\bar x_2$, where $g$ is the microscope magnification. 
We consider the task of estimating the parameter vector $\theta=\{x_1,x_2\}$, from which $\bar{x}_1$ and $\bar{x}_2$ can easily be found, see Fig.~\ref{fig:micro}. 

The localization efficiency is characterized by the tree measures introduced in Sec.~\ref{sec:Efficiency}. Operationally, this can be realized by the following protocol: (i) the verifier prepares two sources and passes them to the microscopist communicating the brightnesses $\mu=\{\mu_1,\mu_2\}$ but not the locations $\theta=\{x_1,x_2\}$; (ii) the microscopist returns the estimates of the locations $\hat\theta=\{\hat x_1,\hat x_2\}$ to the verifier; (iii) such an exchange is repeated many times generating a set of data; (iv) the verifier calculates the covariance matrix $\Covel(\hat\theta)_{jk}=\langle(\hat x_j-x_j)(\hat x_k-x_k)\rangle$ for the subset of data with given $\theta$ and $\mu$; (v)  from the covariance matrix, the verifier calculates the localization efficiency $H(\theta)$ as a function of $\theta$ and $\mu$.

We consider an ordinary wide-field microscope with a circular aperture having the properties of translation and rotation invariance within its field of view. When the number of photons on the image plane within the coherence time is much less than unity (which means that the source is ``dim''), the wave function of the photon coming from the $k$th source is \cite{Tsang16,Tsang17} $\psi(x-x_k)$, where $\psi(x)$ is the impulse-response function of the optical system of the microscope \cite{Goodman-book}. For simplicity, we assume that this function is real and model it with a Gaussian function so that its square, the PSF of the microscope, is normalized to unity:
\begin{equation}
\psi(x) = \frac1{(2\pi\sigma^2)^{1/4}} e^{-x^2/4\sigma^2},
\end{equation}
where the standard deviation $\sigma$ determines the size of each spot.

When just the $k$th source is active (blinking scenario), the state of the probe (single-photon field) on the image plane is $\rho_{Pk}=|\Psi(x_k)\rangle\langle \Psi(x_k)|$, where
\begin{equation}\label{Psix}
|\Psi(x_k)\rangle = \int\limits_{-\infty}^{+\infty}dx \psi(x-x_k)a^\dagger(x)|\text{vac}\rangle,
\end{equation}
$a^\dagger(x)$ being the photon creation operator at point $x$. In direct detection, the POVM for detecting a photon at position $x$ is the projector on the one-photon state at position $x$, i.e., $\Lambda(x)=a^\dagger(x)|\text{vac}\rangle\langle\text{vac}|a(x)$. Then, the probability of measurement outcome $x$ is $p_k(x|x_1,x_2)=\Tr\left[\rho_{Pk}\Lambda(x)\right]=\psi^2(x-x_k)$, where we have used the commutation relation $[a(x),a^\dagger(x')]=\delta(x-x')$. The corresponding Fisher information matrix is $F^{[k]}_{ij}=\sigma^{-2}\delta_{ij}\delta_{ik}$. The  Fisher information matrix of the entire estimation session with $N$ detected photons, $\mu_1N$ of which coming from the first source and $\mu_2N$ coming from the second one, is thus  
\begin{equation}\label{Fblink1D}
\dirprod{\mathbf{F}}_\text{blink} = \frac{N}{2\sigma^2}\left(\begin{array}{cc}
    1+\delta & 0 \\
    0 & 1-\delta
\end{array}\right),    
\end{equation}
where we express the relative brightnesses, satisfying the relation $\mu_1+\mu_2=1$, through one parameter $0<\delta<1$, which we assume to be known, as  $\mu_1=\frac12(1+\delta)$ and $\mu_2=\frac12(1-\delta)$. The average Fisher information in this scenario is $\bar{F}_\text{blink} = N/2\sigma^2$. 

When the two sources are active simultaneously (cofluorescence scenario), the state of the probe on the image plane is $\tilde\rho_{P}=\mu_1\tilde\rho_{P1}+\mu_2\tilde\rho_{P2}$ and the probability of the measurement outcome $x$ is
\begin{equation}\label{p1D}
p(x|x_1,x_2)=\mu_1\psi^2(x-x_1) + \mu_2\psi^2(x-x_2).
\end{equation}
This function is shown in Fig. \ref{fig:1D}. 

\begin{figure}[!ht]
\centering
\includegraphics[width=\linewidth]{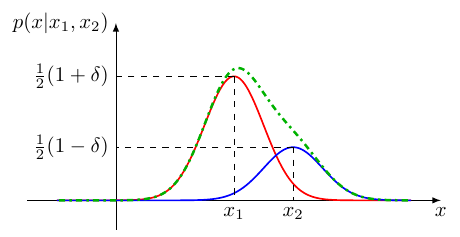}
\caption{Statistical model for two cofluorescent incoherent sources. The solid lines show the point-spread functions centered at positions $x_1$ and $x_2$ on the image plane of the microscope. The dash-dotted line shows their sum corresponding to one observed spot when the sources are active simultaneously.\label{fig:1D}}
\end{figure}

The Fisher information matrix of the entire estimation session with $N$ detected photons distributed according to Eq. (\ref{p1D}) is 
\begin{equation}\label{Fcoflu1D}
\dirprod{F}^\text{coflu}_{ij} = \frac{N\mu_i\mu_j}{\sigma^4} \int\limits_{-\infty}^{+\infty} (x-x_i)(x-x_j)\frac{p_i(x|x_i)p_j(x|x_j)}{p(x|x_1,x_2)}dx.
\end{equation}
\begin{figure*}[!ht]
\centering
\includegraphics[width=0.49\linewidth]{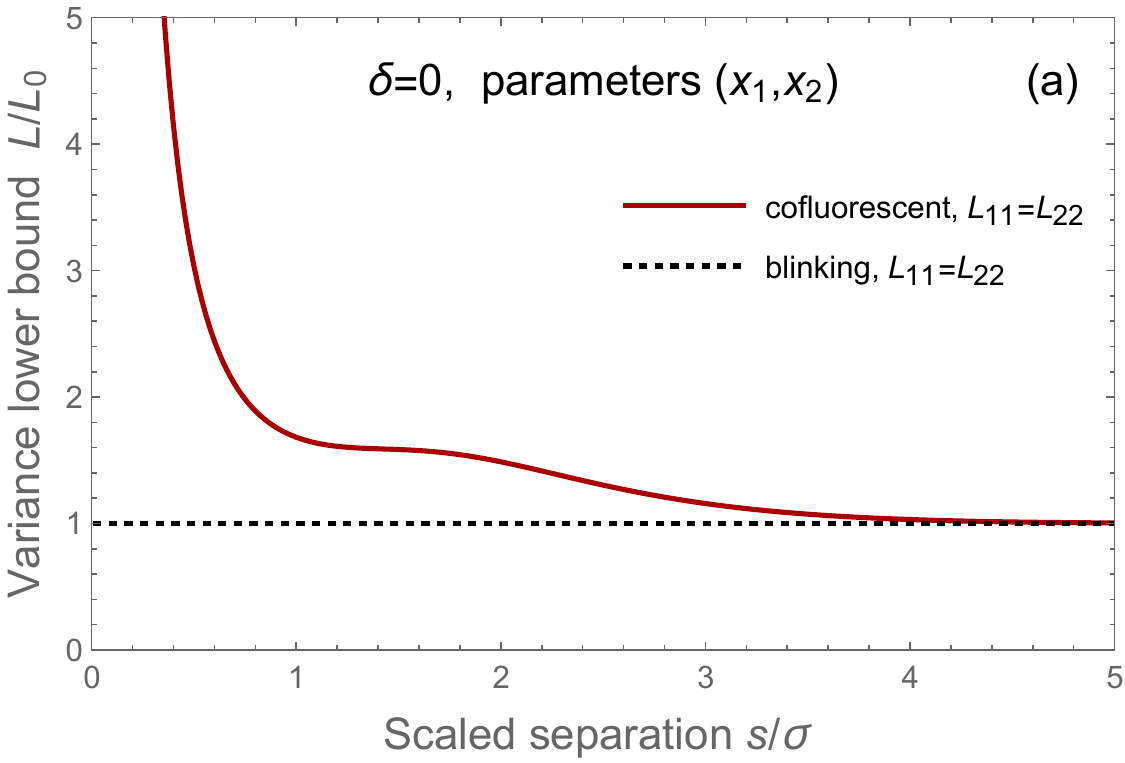}
\includegraphics[width=0.49\linewidth]{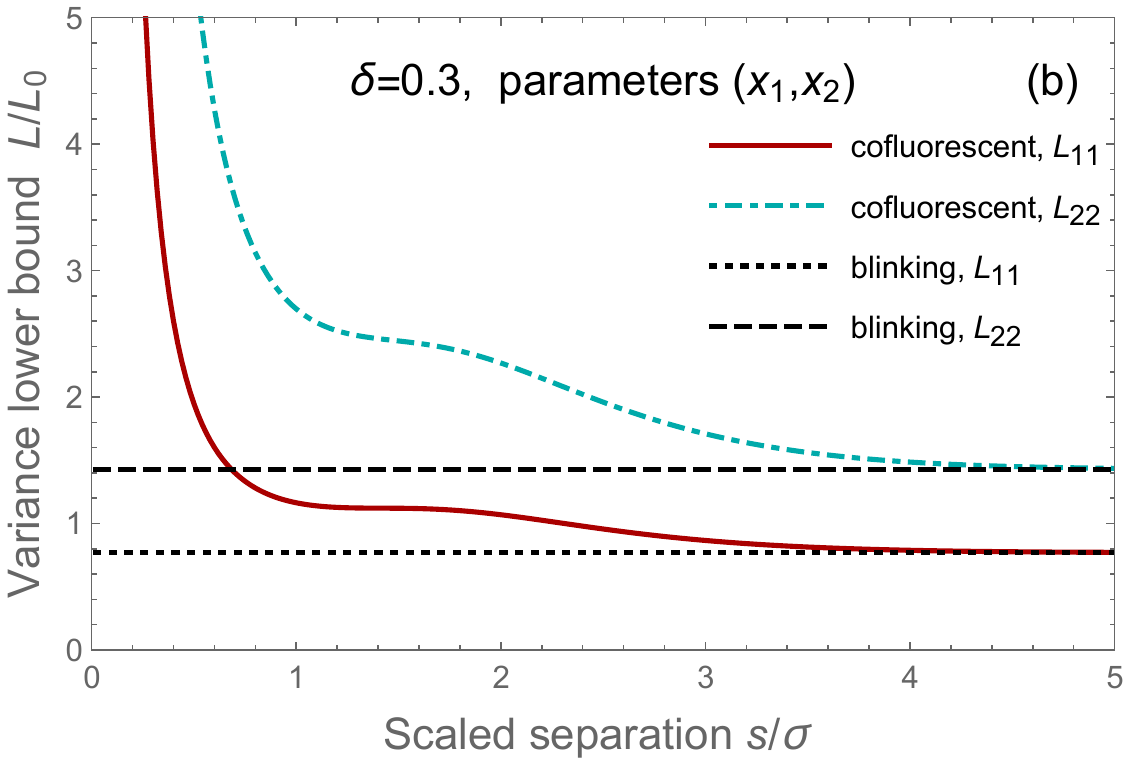}
\includegraphics[width=0.49\linewidth]{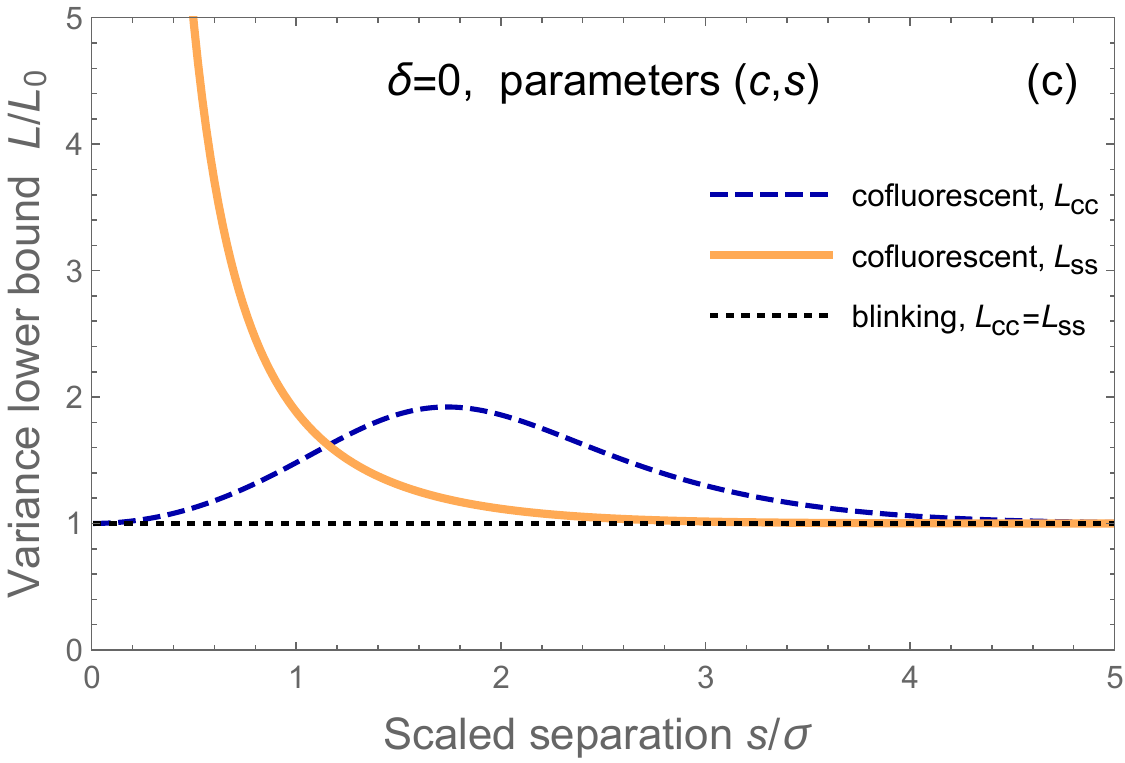}
\includegraphics[width=0.49\linewidth]{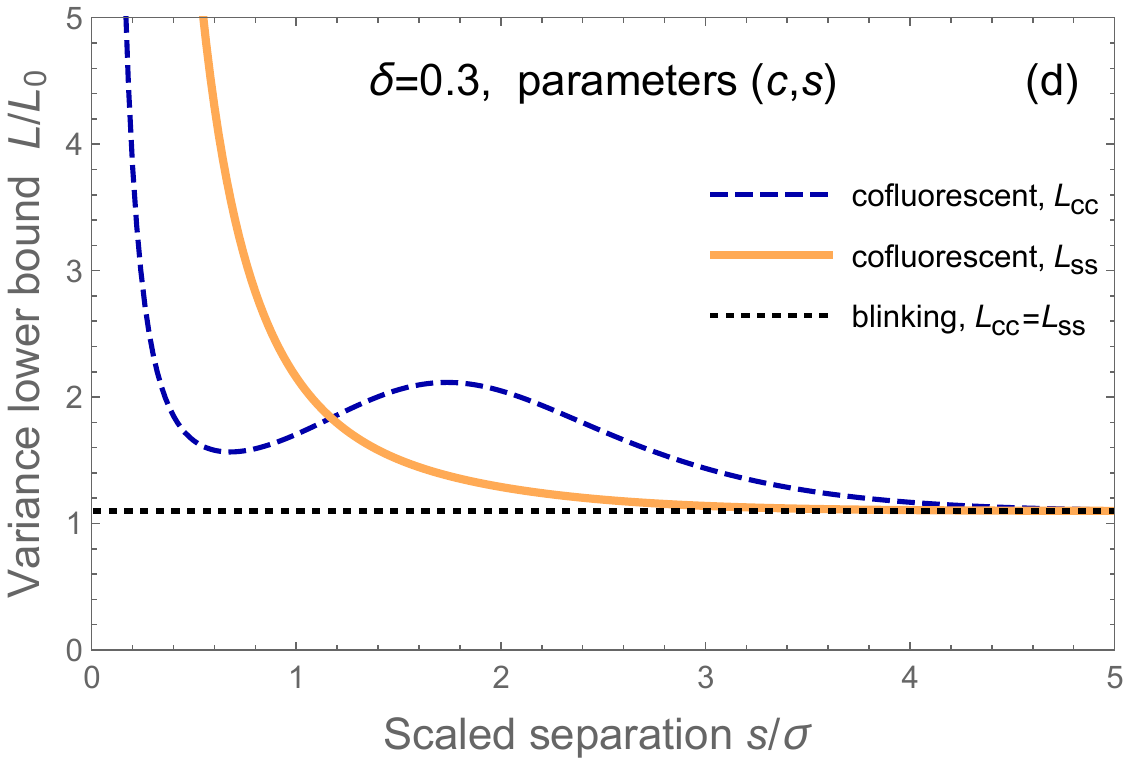}
\includegraphics[width=0.49\linewidth]{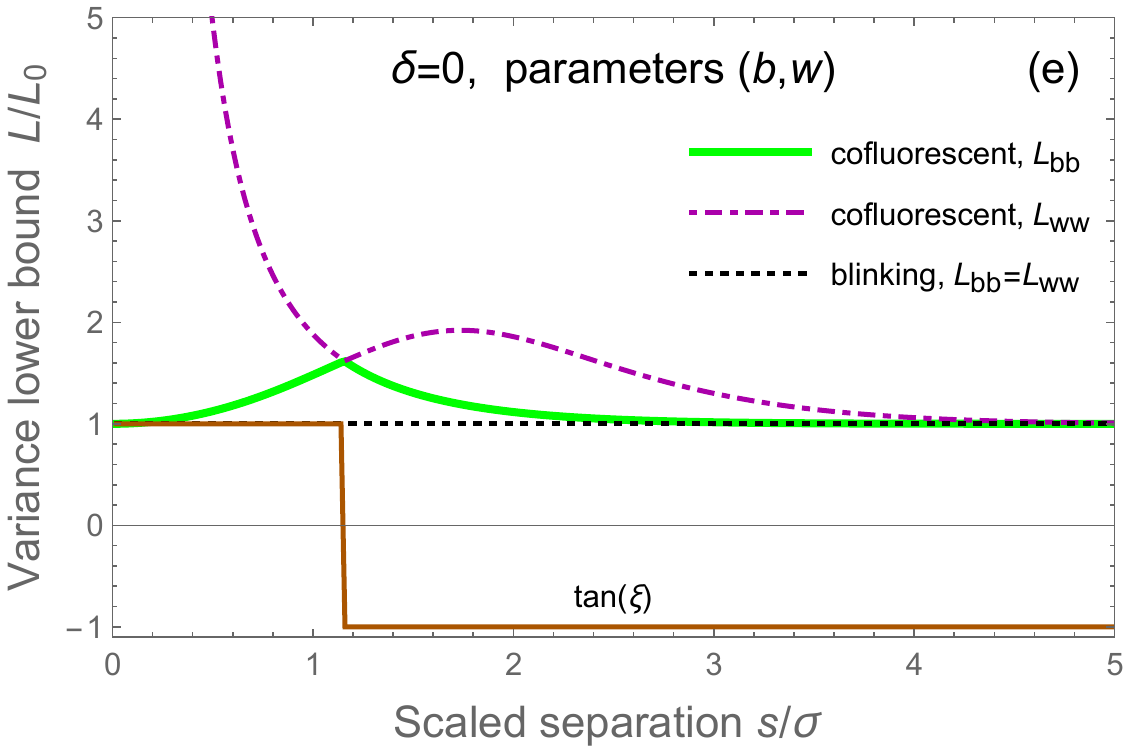}
\includegraphics[width=0.49\linewidth]{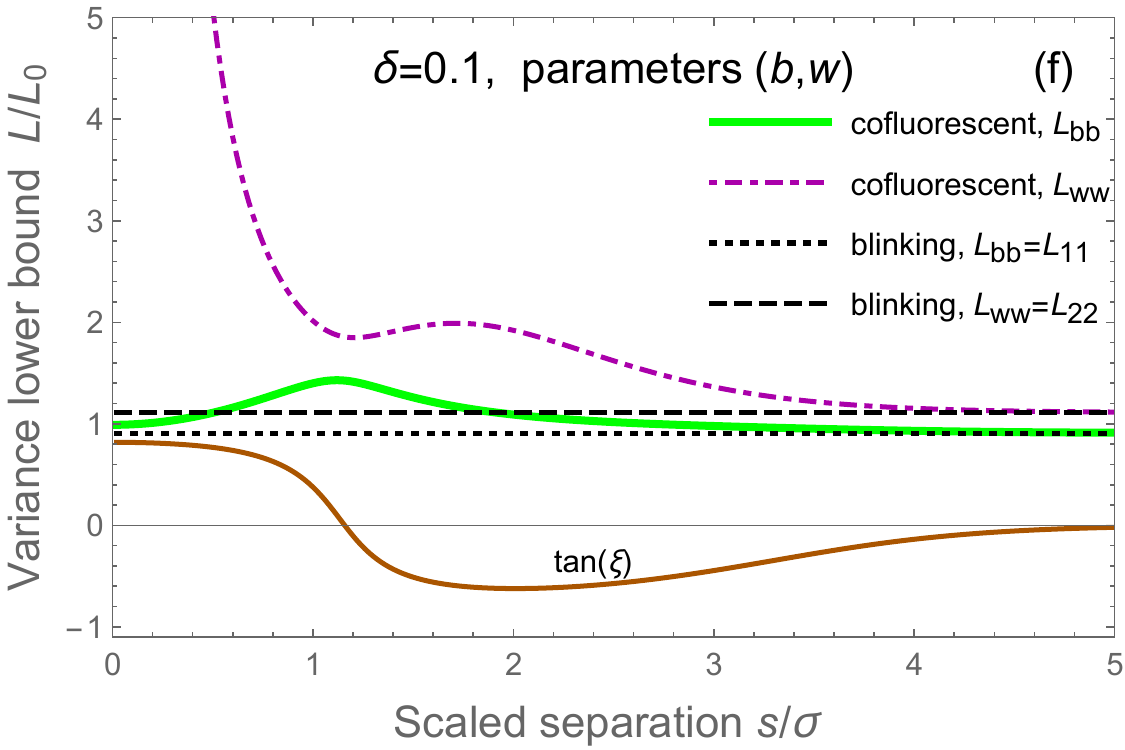}
\caption{Variance lower bounds in the units of $L_0=2\sigma^2/N$ for estimating the unknown positions $(x_1,x_2)$ of two incoherent sources by detecting $N$ photons in a microscope with the point-spread function of standard deviation $\sigma$. $\delta$ is the relative brightness difference of the sources. Parameters $(c,s)$ are the scaled centroid and separation, respectively. Parameters $(b,w)$ are the eigenparameters of the Fisher information matrix obtained from $(x_1,x_2)$ via a rotation by angle $\xi$: $b=x_1\cos\xi + x_2\sin\xi $ and $w= x_2\cos\xi - x_1\sin\xi$. \label{fig:L}}
\end{figure*}

The corresponding variance lower bounds $L_{ii}=(\dirprod{\mathbf{F}}^{-1})_{ii}$ are calculated numerically and shown in Fig. \ref{fig:L}(a) for the case of equal brightnesses $\mu_1=\mu_2=\frac12$. In this case, due to the symmetry with respect to the exchange of parameters $x_1$ and $x_2$, we have $L_{11}=L_{22}$ for both scenarios, and the variance of the cofluorescence scenario is always above that of the blinking scenario, in correspondence with Theorem \ref{theorem:Main}.

In Fig. \ref{fig:L}(b), we see a similar situation for the case of unequal brightnesses. The variance lower bounds are different in this case, with the brighter source having a lower position variance. In the blinking scenario, the variance lower bounds for the positions of the first and second sources are $(1+\delta)^{-1}L_0$ and $(1-\delta)^{-1}L_0$ respectively, where $L_0=1/\bar{F}_\text{blink}$. In the cofluorescence scenario, the variances of both positions lie above those of the blinking scenario, though the case of unequal brightnesses is not covered by Theorem \ref{theorem:Main} for individual bounds.

We also see in Fig. \ref{fig:L}(a) that the variances of both positions tend to infinity at a small separation, which means that these positions are difficult to estimate. However, intuitively we understand that the centroid $(x_1+x_2)/2$ of the overlapping spots can still be estimated with high precision. To show it explicitly, we pass from our parameters $\theta=\{x_1,x_2\}$ to a rotated vector of parameters $\varphi=\{c,s\}$, where $c=(x_1+x_2)/\sqrt{2}$ is the scaled centroid and $s=(x_2-x_1)/\sqrt{2}$ is the scaled separation. This rotation in the parameter space can be written as $\varphi=\mathbf{O}_{\pi/4}^T\theta$, where
\begin{equation}
\mathbf{O}_\xi = \left(\begin{array}{cc}
    \cos\xi & -\sin\xi \\
    \sin\xi & \cos\xi
\end{array}\right)
\end{equation}
is the matrix of orthogonal rotation by angle $\xi$. The Fisher information matrix in the new basis is simply  $\tilde{\dirprod{\mathbf{F}}} =\mathbf{O}_{\pi/4}^T\dirprod{\mathbf{F}}\mathbf{O}_{\pi/4}$. Applying this transformation to Eq. (\ref{Fblink1D}), we obtain
\begin{equation}\label{tildeFblink1D}
\tilde{\dirprod{\mathbf{F}}}_\text{blink} = \frac{N}{2\sigma^2}\left(\begin{array}{cc}
    1 & -\delta \\
    -\delta & 1
\end{array}\right),    
\end{equation}
which gives the variance lower bounds $L_{cc}=L_{ss}=L_0/(1-\delta^2)$. In the cofluorescence scenario, the variance lower bounds for $c$ and $s$ are calculated numerically and shown in Fig. \ref{fig:L}(c) for the case of equal brightnesses. We see that, indeed, the variance of the scaled centroid at zero separation approaches that of blinking sources, but the variance of the scaled separation goes to infinity. Now, we can explain why the positions $x_1=(c-s)/\sqrt{2}$ and $x_2=(c+s)/\sqrt{2}$ become completely indeterminate in the $s\to0$ limit: they include the scaled separation $s$ whose variance tends to infinity. 

\begin{figure*}[!ht]
\centering
\includegraphics[width=0.49\linewidth]
{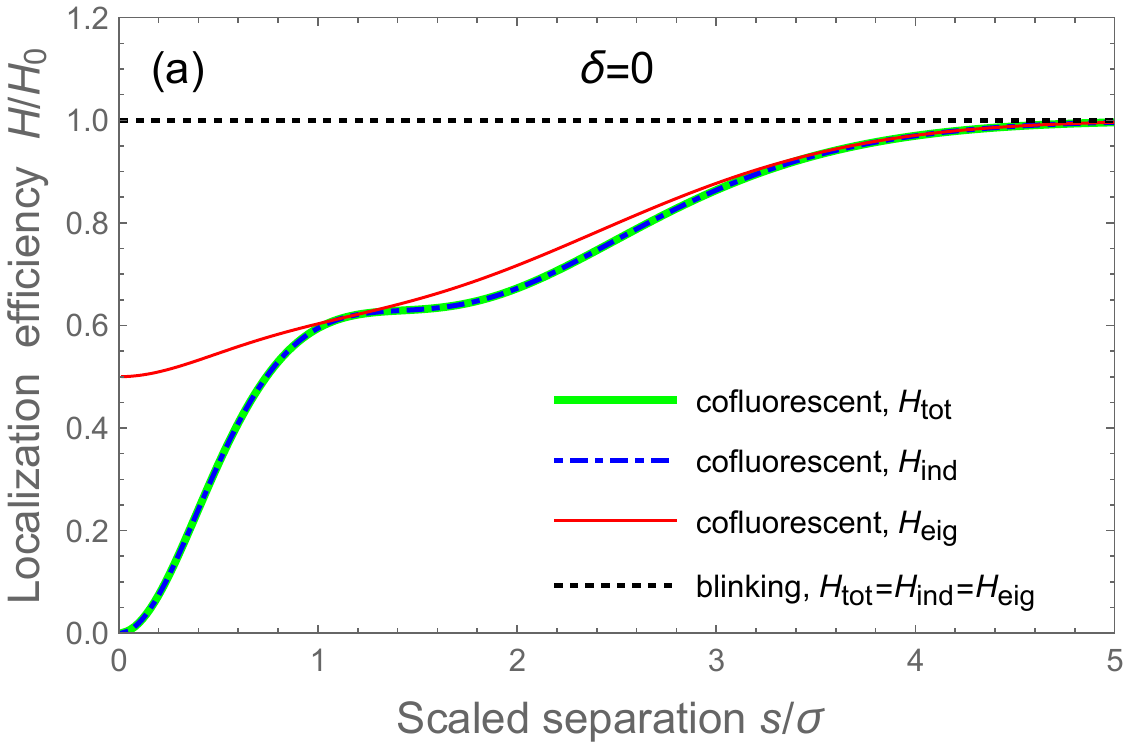}
\includegraphics[width=0.49\linewidth]
{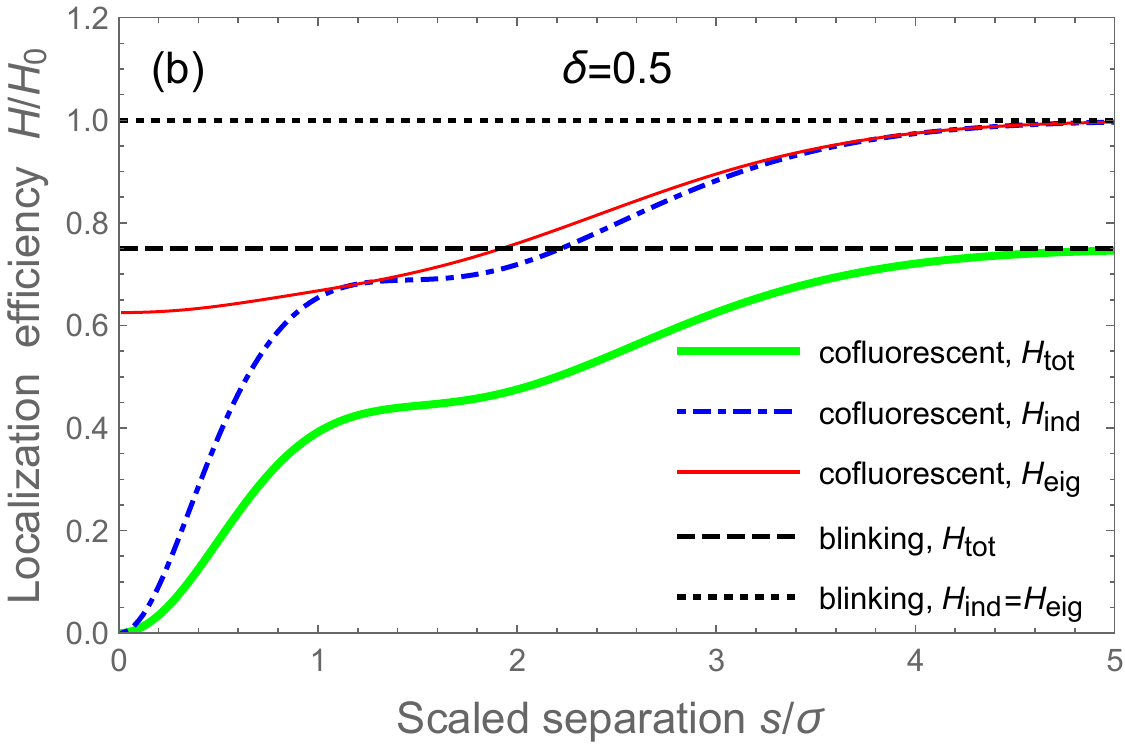}
\caption{Localization efficiency in the units of $H_0=N/2\sigma^2$ for estimating the unknown positions $(x_1,x_2)$ of two incoherent sources by detecting $N$ photons in a microscope with the point-spread function of standard deviation $\sigma$. The blinking scenario demonstrates a higher localization efficiency, whatever measure is used for the latter. \label{fig:EfficiencyF}}
\end{figure*}

The case of unequal brightnesses is illustrated in Fig. \ref{fig:L}(d). We see that both $c$ and $s$ become indeterminate at $s\to0$ in the co-fluorescence scenario. This may be caused by the presence of one variable with a divergent variance in both these parameters, as is the case for $x_1$ and $x_2$ analyzed above. Indeed, the limiting value of the Fisher information matrix is 
\begin{equation}\label{Fblink1Dlim}
\lim_{x_2\to x_1}\dirprod{\mathbf{F}}_\text{coflu} = \frac{N}{\sigma^2}\left(\begin{array}{cc}
    \mu_1^2 & \mu_1\mu_2 \\
    \mu_1\mu_2 & \mu_2^2
\end{array}\right).    
\end{equation}
This matrix has a zero determinant, but its trace is non-zero. It means that one of its eigenvalues is zero, but another is not: $\lambda_1=N(\mu_1^2+\mu_2^2)/\sigma^2$, $\lambda_2=0$. As a consequence, at any $s$, we can introduce two  eigenparameters, as discussed in Sec. \ref{sec:Efficiency}: the best combination of the two positions $b=x_1\cos\xi + x_2\sin\xi $ corresponding to the higher eigenvalue of $\dirprod{\mathbf{F}}_\text{coflu}(s)$ and the weighted separation $w= x_2\cos\xi - x_1\sin\xi$ corresponding to its lower eigenvalue. The latter parameter has a divergent variance at $s\to0$, while the former is well determined.

In Fig.~\ref{fig:L}(e), we illustrate the behavior of eigenparameter variances and the angle $\xi$ in the case of equal brightnesses. We see that $\xi=\pi/4$ at low separation, which corresponds to $b=c$ and $w=s$, exactly as in Fig.~\ref{fig:L}(c). At higher separation, the angle becomes $\xi=-\pi/4$, which means that the parameters are relabeled: $b=-s$ and $w=c$. This happens because $s$ has a higher eigenvalue in this region. We also see that both eigenvalues of $\dirprod{\mathbf{F}}_\text{coflu}^{-1}(\theta)$ are larger than the degenerate eigenvalue of $\dirprod{\mathbf{F}}_\text{blink}^{-1}(\theta)$, as prescribed by Eq.~(\ref{Main}).

A more interesting case of unequal brightnesses is illustrated in Fig.~\ref{fig:L}(f). We see that the angle $\xi$ changes continuously, tending to zero at high separation, where the matrix $\dirprod{\mathbf{F}}_\text{coflu}(s)$ becomes diagonal in the $(x_1,x_2)$ basis. The parameter $b$ has always a finite variance, while that of $w$ is divergent at a small separation. Both variances in the cofluorescent scenario lie above their counterparts in the blinking scenario, which means that $\lambda_i(\dirprod{\mathbf{F}}_\text{blink}) \ge\lambda_i(\dirprod{\mathbf{F}}_\text{coflu})$, where $\lambda_i(\mathbf{A})$ is the $i$th eigenvalue of a Hermitian matrix $\mathbf{A}$, when the eigenvalues are sorted in a nonincreasing order. It is known that for any two Hermitian matrices $\mathbf{A}$ and $\mathbf{B}$, $\mathbf{A}\ge \mathbf{B}$ is a sufficient but not necessary condition for $\lambda_i(\mathbf{A})\ge \lambda_i(\mathbf{B})$ (see Corollary 7.7.4. in Ref.~\cite{HornJohnson}). This means that we cannot conclude from Fig.~\ref{fig:L}(f) that inequality (\ref{Main}) holds also in the case of unequal brightnesses, not comprised by Theorem~\ref{theorem:Main}. However, an additional numerical study shows that the minimal eigenvalue of the matrix $\dirprod{\mathbf{F}}_\text{blink}-\dirprod{\mathbf{F}}_\text{coflu}$ is positive in the limits of computational error, which is a numerical evidence in favor of validity of inequality (\ref{Main}) in a more general case of unequal brightnesses.

Thus, the variance lower bounds for eigenparameters $L_{bb}$ and $L_{ww}$ appear to be the best characteristics for the overall localization efficiency. If we wish to combine them in one scalar figure-of-merit, we should avoid the divergence of $L_{ww}$ at zero separation. The best way to do it is to pass from variances to precisions, as discussed in Sec. \ref{sec:Efficiency}. The localization efficiency defined in three possible ways is shown in Fig.~\ref{fig:EfficiencyF}(a) for the case of equal brightnesses and saturated Cram\'er-Rao bound. This figure confirms what we have already seen in Fig.~\ref{fig:L}(e): the Fisher information matrix of blinking sources is larger than that of cofluorescent ones and, therefore, all three measures of localization efficiency are higher for blinking sources.

The case of unequal brightnesses is shown in Fig.~\ref{fig:EfficiencyF}(b). We see that $H_\text{eig}$ is higher for blinking sources than for cofluorescent ones, which is a direct consequence of the trace convexity of the Fisher information matrix, proven in Theorem \ref{theorem:convexity}. The other two measures, $H_\text{tot}$ and $H_\text{ind}$, are also higher for blinking sources, though this fact lies outside the scope of the validity of Theorem \ref{theorem:Main}. We also clearly see the ordering of the measures prescribed by Eq.~(\ref{ordering}).

\section{Extension to quantum Fisher information \label{sec:Extension}}

The results of the previous sections can be rather straightforwardly extended to the case where the upper bound for localization precision is quantified by quantum Fisher information, which may correspond to a better measurement than the direct detection, as discussed in Sec. \ref{sec:QFI}.

\subsection{Additivity and convexity of quantum Fisher information matrix}

We begin by reformulating the definition of quantum Fisher information matrix \cite{Helstrom-book,HolevoBook82,Paris09} as a map from a family of density operators in the state space of a quantum system to the space of real symmetric matrices.

\begin{definition}\label{definition:Q}
Given a positive integer $M$, a state space of a quantum system $\mathcal{S}$, and a family of density operators $\rho(\theta)\in\mathcal{S}$ depending on parameter $\theta\in\mathbb{R}^M$, the quantum Fisher information matrix for the information contained in $\rho(\theta)$ about $\theta$ is the $M\times M$-matrix-valued map $\bm{\mathcal{Q}}_M^{(\mathcal{S})}\left[\rho(\theta)\right]$, whose matrix elements are
\begin{equation}
\mathcal{Q}_M^{(\mathcal{S})}\left[\rho(\theta)\right]_{ij} = \frac12\Tr\left\{\left(\mathcal{L}_i\mathcal{L}_j+\mathcal{L}_j\mathcal{L}_i\right)\rho(\theta)\right\},
\end{equation}
where the symmetric logarithmic derivative $\mathcal{L}_i$ is defined by the equation
\begin{equation}\label{SLD2}
\frac{\partial}{\partial\theta_i} \rho(\theta) = \frac12\left[\mathcal{L}_i\rho(\theta) + \rho(\theta)\mathcal{L}_i\right]. 
\end{equation}
\end{definition}

Now, we formulate the property of additivity for this map, which is well known from the literature \cite{Helstrom-book,HolevoBook82,Paris09}.

\begin{lemma}\label{lemma:additivityQ}
The quantum Fisher information matrix possesses the property of additivity: 
\begin{eqnarray}
&&\bm{\mathcal{Q}}_M^{(\mathcal{S}_1\otimes\mathcal{S}_2)} \left[\rho_1(\theta)\otimes\rho_2(\theta)\right]\\\nonumber
&&=\bm{\mathcal{Q}}_M^{(\mathcal{S}_1)}\left[\rho_1(\theta)\right] +\bm{\mathcal{Q}}_M^{(\mathcal{S}_2)}\left[\rho_2(\theta)\right],
\end{eqnarray}
where $\rho_1(\theta)\in\mathcal{S}_1$ and $\rho_2(\theta)\in\mathcal{S}_2$.
\end{lemma}

Next, we formulate the property of convexity in the single-parameter case. The property of non-strict convexity is known from the literature~\cite{Toth13,Toth22}, but we need the strict convexity property, which is proven below. In the proof, we use the well-known fact that, in the one-parameter case,  quantum Fisher information represents the result of maximizing classical Fisher information over all quantum measurements \cite{HolevoBook82,Paris09}. 

\begin{lemma}\label{lemma:convexityQ}
In the single-parameter case, $M=1$, the quantum Fisher information matrix (reducing to a scalar in this case) possesses the property of strict convexity:
\begin{eqnarray}
     &&\mathcal{Q}_1^{(\mathcal{S})}\left[\gamma \rho_1(\theta)+(1-\gamma)\rho_2(\theta)\right] \\\nonumber
     &&< \gamma\mathcal{Q}_1^{(\mathcal{S})}\left[\rho_1(\theta)\right]+(1-\gamma)\mathcal{Q}_1^{(\mathcal{S})}\left[\rho_2(\theta)\right], 
\end{eqnarray}
for any $0<\gamma<1$ and $\rho_1(\theta)$ and $\rho_2(\theta)$ satisfying the requirement
\begin{equation}\label{requirement}
\Tr\{\Lambda_\gamma(X)\rho_1(\theta)\} \not\equiv \Tr\{\Lambda_\gamma(X)\rho_2(\theta)\},  
\end{equation}
where $\Lambda_\gamma(X)$ is the POVM maximizing the information retrieved from the state $\gamma \rho_1(\theta) + (1 - \gamma) \rho_2(\theta)$ by measuring an $\ell$-dimensional variable $X$.
\end{lemma}
\begin{proof}
The left-hand side of Eq.~(47) can be represented as
\begin{multline}
    \label{eq:QFI to FI}
    \mathcal{Q}_1^{(\mathcal{S})}\left[ \gamma \rho_1(\theta) + (1 - \gamma) \rho_2(\theta) \right] \\ = \mathcal{F}_1^{(\ell)}\left[ \gamma q_1(X|\theta) + (1 - \gamma) q_2(X|\theta) \right],
\end{multline}
where $\mathcal{F}_1^{(\ell)}$ is the classical Fisher information defined in Definition~\ref{definition:F}, and the probability distributions $q_i(X|\theta) = \Tr\left\{\Lambda_\gamma(X) \rho_i(\theta)\right\}$ correspond to the optimal measurement of $\gamma \rho_1(\theta) + (1 - \gamma) \rho_2(\theta)$.

The strict convexity property of the classical Fisher information (Lemma~\ref{lemma:convexity}) and Eq.~(\ref{requirement}) imply that
\begin{multline}
    \label{eq:FI inequality for QFI}
    \mathcal{F}_1^{(\ell)}\left[ \gamma q_1(X|\theta) + (1 - \gamma) q_2(X|\theta) \right] \\ < \gamma \mathcal{F}_1^{(\ell)}\left[ q_1(X|\theta) \right] + (1 - \gamma) \mathcal{F}_1^{(\ell)}\left[q_2(X|\theta) \right].
\end{multline}
For each term in the right-hand side of inequality (\ref{eq:FI inequality for QFI}), the following relation holds:
\begin{equation}
    \label{eq:FI to QFI}
    \mathcal{F}_1^{(\ell)}\left[ q_i(X|\theta) \right] \le \mathcal{F}_1^{(\ell)}\left[ q_i'(X|\theta) \right] = \mathcal{Q}_1^{(\mathcal{S})}\left[ \rho_i(\theta) \right],
\end{equation}
where the probability distribution $q_i'(X|\theta) = \Tr\left\{\Lambda_i(X) \rho_i(\theta)\right\}$, $i = 1,2$, corresponds to the measurement $\Lambda_i(X)$ maximizing the information retrieved from the state $\rho_i(\theta)$ and, therefore, defines the quantum Fisher information for that state.

Combining Eqs.~(\ref{eq:QFI to FI})--(\ref{eq:FI to QFI}), we prove the lemma.
\end{proof}

This property can be extended to the multiple-parameter case by the following theorem.

\begin{theorem}\label{theorem:convexityQ}
Given two positive integers $K\ge2$ and $d$, in the general case of $M$ parameters split into $K$ groups of $d$ parameters each, $\theta=\{\bar\theta_1,...,\bar\theta_K\}$, where $\bar\theta_k\in\mathbb{R}^d$, $M=Kd$, and given real numbers $\mu_1$,...,$\mu_K$ belonging to $[0,1)$ and summing up to unity, the quantum Fisher information matrix possesses the following properties. 
\begin{enumerate}
\item[A.] Strict weak matrix convexity in the case of equal weights $\mu_k=1/K$:
\begin{equation}\label{convexityQ}    \bm{\mathcal{Q}}_M^{(\mathcal{S})}\left[\frac1K\sum\limits_{k=1}^{K} \rho_k(\theta)\right] < \frac1K\sum\limits_{k=1}^{K} \bm{\mathcal{Q}}_M^{(\mathcal{S})}\left[\rho_k(\theta)\right],
\end{equation}
under condition that the density operator $\rho_k(\theta)$ depends on $\bar\theta_k$ only and 
\begin{equation}\label{conditionQ}
    \mathcal{Q}_M^{(\mathcal{S})} \left[\rho_k(\theta)\right]_{ij} = \left\{ \begin{array}{ll}
       q_0,  &  i=j\in[(k-1)d+1,kd],\\
       0,  & \mathrm{otherwise},
    \end{array}\right.
\end{equation}
where $q_0$ is a positive number independent of $k$.
\item[B.] Strict trace convexity in the general case of arbitrary weights $\mu_k$:
\begin{equation}\label{trace-convexityQ}    
\Tr\left\{\bm{\mathcal{Q}}_M^{(\mathcal{S})} \left[\sum\limits_{k=1}^{K} \mu_k \rho_k(\theta)\right]\right\} < \sum\limits_{k=1}^{K}\mu_k \Tr\left\{\bm{\mathcal{Q}}_M^{(\mathcal{S})} \left[\rho_k(\theta)\right]\right\}.
\end{equation}
\end{enumerate}
\end{theorem}

This theorem can be proven exactly the same as Theorem~\ref{theorem:convexity}, by rotating the parameter vector to the eigenparameters of the quantum Fisher information matrix and applying the single-parameter convexity, established by Lemma~\ref{lemma:convexityQ} to each diagonal entry. In doing this, we assume that the condition (\ref{requirement}) is satisfied every time convexity is used for a mixture of two density operators. If, by chance, this condition is not satisfied, then the sign ``less'' should be replaced by ``less or equal.'' Note that in the task of localization of sources of light, considered here, density operators $\rho_1(\theta)$ and $\rho_2(\theta)$ typically depend on different sets of parameters -- the coordinates of the corresponding sources -- and the full equivalence of the conditional probability distributions resulting from them is highly unlikely.

\subsection{Quantum Fisher information for localization of blinking sources \label{sec:LocalizationQ}}

The consideration of Sec.~\ref{sec:Localization} can be extended to quantum Fisher information. In the cofluorescence scenario, the state of the probe, defined in Sec.~\ref{sec:CRB}, is $\tilde{\dirprod{\rho}}_P^\text{coflu}=\left(\tilde\rho_P^\text{coflu}\right)^{\otimes N}$, where
\begin{equation}\label{rhomixprobe}
\tilde\rho_P^\text{coflu}=\sum_k\mu_k\tilde\rho_{Pk}, 
\end{equation}
and the quantum Fisher information matrix, after Lemma~\ref{lemma:additivityQ}, has the form $\dirprod{\mathbf{Q}}_\text{coflu}(\theta)=N\mathbf{Q}_\text{coflu}(\theta)$, where $\mathbf{Q}_\text{coflu}(\theta)$ is the single-sample quantum Fisher information matrix with the elements
\begin{equation}\label{Qijsample}
Q_{\text{coflu},ij}(\theta) = \frac12\Tr\left\{\left(\mathcal{L}_i\mathcal{L}_j+\mathcal{L}_j\mathcal{L}_i\right)\tilde\rho_P^\text{coflu}\right\}.
\end{equation}

In the blinking scenario, the state of the probe is 
\begin{equation}\label{rhoblinkprobe}
\tilde{\dirprod{\rho}}_P^\text{blink} = \tilde\rho_{P1}^{\otimes N_1}\otimes...\otimes\tilde\rho_{PK}^{\otimes N_K}
\end{equation}
and the quantum Fisher information matrix is 
$\dirprod{\mathbf{Q}}_\text{blink}(\theta) =\sum_kN_k\mathbf{Q}^{[k]}(\theta)$, where $\mathbf{Q}^{[k]}(\theta)$ is the single-sample quantum Fisher information matrix in the $k$th window, where just the $k$th source is active:
\begin{equation}\label{Qijksample2}
Q_{ij}^{[k]}(\theta) = \frac12\Tr\left\{\left(\mathcal{L}_i\mathcal{L}_j+\mathcal{L}_j\mathcal{L}_i\right)\tilde\rho_{Pk}\right\}.
\end{equation}
In both Eqs.~(\ref{Qijsample}) and (\ref{Qijksample2}) the symmetric logarithmic derivative is defined for the respective density operator according to Eq.~(\ref{SLD}).

The advantage of a blinking scenario is formulated by the following theorem. For its formulation, we need one more definition.

\begin{definition}\label{definition:Qinv}
Microscopy is quantum translation and rotation invariant if the quantum Fisher information matrix of the state $\tilde\rho_P(r)$ of a single-photon field on the image plane about the coordinates $r=\{x,y\}$ of the image of a single source having emitted this photon satisfies condition 
$\mathcal{Q}_2^{(\mathcal{S})}\left[\tilde\rho_P(r)\right]_{ij}=q_0\delta_{ij}$, where $q_0$ is independent of $r$. 
\end{definition}
\begin{theorem}\label{theorem:MainQ}
For $\dirprod{\mathbf{Q}}_\text{coflu}(\theta)$ defined with the object state (\ref{rhomixprobe}) and $\dirprod{\mathbf{Q}}_\text{blink}(\theta)$ defined with the object state (\ref{rhoblinkprobe}), the following inequalities hold:
\begin{enumerate}
    \item[A.] In the case of quantum translation and rotation invariant microscopy with $N_k=N/K$,
    \begin{equation}\label{MainQ}
    \dirprod{\mathbf{Q}}_\text{blink}(\theta) > \dirprod{\mathbf{Q}}_\text{coflu}(\theta).
    \end{equation}
    \item[B.] In the general case of arbitrary $N_k$, at least two of which are non-zero,
    \begin{equation}\label{MainTraceQ}
    \Tr\{\dirprod{\mathbf{Q}}_\text{blink}(\theta)\} > \Tr\{\dirprod{\mathbf{Q}}_\text{coflu}(\theta)\}.
    \end{equation}
\end{enumerate}
\end{theorem}

The proof of this theorem follows the lines of that of Theorem~\ref{theorem:Main}: Applying the additivity and convexity properties of Lemma~\ref{lemma:additivityQ} and Theorem~\ref{theorem:convexityQ}.

As a corollary to Theorem~\ref{theorem:MainQ}, we have $\dirprod{\mathbf{Q}}_\text{blink}^{-1}(\theta) < \dirprod{\mathbf{Q}}_\text{coflu}^{-1}(\theta)$ in case A. Does it mean that the localization efficiency is always higher in the blinking scenario for any choice of measurement? In the case of Fisher information discussed in Sec.~\ref{sec:Localization}, the answer was positive, because the Cram\'er-Rao bound is always asymptotically saturable. The quantum Cram\'er-Rao bound, however, is asymptotically saturable only in the case where the equality
\begin{equation}\label{commutativity}
\Tr\left\{
\left(\mathcal{L}_i\mathcal{L}_j-\mathcal{L}_j \mathcal{L}_i\right)\tilde{\dirprod{\rho}}_P(\theta)
\right\}=0   
\end{equation}
holds for any $i$ and $j$ meaning that the parameters $\theta_i$ and $\theta_j$ are compatible for a simultaneous quantum estimation \cite{Albarelli20}. Let us show that this equality holds for the blinking scenario in the case of a real impulse-response function $\psi(x,y)$. 

When just the $k$th source is active, the state of the single-photon field on the image plane is $\tilde\rho_{Pk}=|\Psi_k\rangle\langle \Psi_k|$, where, similar to Eq.~(\ref{Psix}), we write
\begin{equation}\label{Psixy}
|\Psi_k\rangle = \iint\limits_{-\infty}^{+\infty}dxdy \psi(x-x_k,y-y_k)a^\dagger(x,y)|\text{vac}\rangle.
\end{equation}
with $a^\dagger(x,y)$ being the photon creation operator at point $(x,y)$. Substituting this density operator into Eq.~(\ref{SLD}), we find the symmetric logarithmic derivatives in the form
\begin{equation}\label{Li}
\mathcal{L}_i = |\partial_i\Psi_k\rangle\langle \Psi_k| + |\Psi_k\rangle\langle\partial_i \Psi_k|,
\end{equation}
where $i$ takes the values $x_k,y_k$ and the states $|\partial_x\Psi_k\rangle$ and $|\partial_y\Psi_k\rangle$ are defined similarly to Eq.~(\ref{Psixy}) but with replacements $\psi(x,y)\to\partial_x\psi(x,y)$ and $\psi(x,y)\to\partial_y\psi(x,y)$, respectively. Note, that the square of the function $\psi(x,y)$ is normalized to unity, but the square of its derivative is not necessarily that, i.e., the state $|\partial_i\Psi_k\rangle$ is not normalized in general. Differentiating the norm $\langle\Psi_k|\Psi_k\rangle=1$, we obtain $\Real\left\{\langle\partial_i\Psi_k|\Psi_k\rangle\right\}=0$, and since the impulse-response function is assumed to be real, $|\partial_i\Psi_k\rangle$ is orthogonal to $|\Psi_k\rangle$. Substituting Eq.~(\ref{Li}) into Eq.~(\ref{commutativity}), we ascertain that the commutator average is zero. It means that $x_k$ and $y_k$ are compatible for a simultaneous quantum estimation in the window where their source is active. The same argument applies to the locations of all other sources. 

In summary, the quantum Cram\'er-Rao bound is asymptotically saturable in the blinking scenario. We do not discuss here whether it is saturable or not in the cofluorescence scenario, because, in any case, the advantage of the blinking scenario, expressed by inequality (\ref{advantage}), holds for the optimal quantum measurement. As in Sec.~\ref{sec:Localization}, this inequality holds for any measure of localization efficiency in the case of equal brightnesses and a quantum translation and rotation invariant microscopy, and for $H_\text{eig}(\hat\theta)$ in the general case.

It is easy to see that the property of quantum translation and rotation invariance, formulated in Definition~\ref{definition:Qinv}, is reached for a rotationally invariant impulse-response function, $\psi(\mathbf{R}r)=\psi(r)$, where $\mathbf{R}$ is a matrix of rotation on the $(x,y)$ plane. Indeed, from Eqs. (\ref{Psixy}) and (\ref{Li}), we obtain $\mathcal{Q}_2^{(\mathcal{S})} \left[\tilde\rho_{Pk}\right]_{ij}=J_{ij}$, where
\begin{equation}
J_{ij} = \langle\partial_i\Psi_k|\partial_j\Psi_k\rangle = \iint\limits_{-\infty}^{+\infty} \partial_i\psi(x,y)\partial_j\psi(x,y)dxdy, 
\end{equation}
which reduces to $J_{ij} = q_0\delta_{ij}$ for a rotationally invariant impulse-response function.

\subsection{Localization in one dimension}
We illustrate the developed theory by the same example of two incoherent sources in one spatial dimension as in Sec.~\ref{sec:Example}. We find easily from the formalism of the preceding section that
\begin{equation}\label{Qblink1D}
\dirprod{\mathbf{Q}}_\text{blink} =\dirprod{\mathbf{F}}_\text{blink} = \frac{N}{2\sigma^2}\left(\begin{array}{cc}
    1+\delta & 0 \\
    0 & 1-\delta
\end{array}\right).    
\end{equation}

The quantum Fisher information matrix in the case of two cofluorescent incoherent sources with unknown relative brightness was found by {\v{R}}eha{\v{c}}ek and coworkers in Ref.~\cite{Rehacek17} for the same Gaussian shape of the impulse-response function. We assume that the relative brightness is known and use from this reference only the $2\times2$ submatrix related to the centroid $S_0=(x_1+x_2)/2$ and the separation $S=x_1-x_2$, which we denote by $\dirprod{\mathbf{Q}}_\text{coflu}^R$. We transform this matrix into our parameters $x_1$ and $x_2$ as $\dirprod{\mathbf{Q}}_\text{coflu} =\mathbf{A}^T\dirprod{\mathbf{Q}}_\text{coflu}^R\mathbf{A}$, where $\mathbf{A}$ is the transformation matrix for the (column) vectors of the parameters, $\{S_0,S\}=\mathbf{A}\{x_1,x_2\}$, and obtain
\begin{equation}\label{Qcoflu1D}
\dirprod{\mathbf{Q}}_\text{coflu} = \frac{N}{2\sigma^2}\left(\begin{array}{cc}
    1+\delta-\beta & -\beta \\
    -\beta & 1-\delta-\beta
\end{array}\right),    
\end{equation}
where
\begin{equation}
\beta=\frac1{8\sigma^2}(1-\delta^2)(x_1-x_2)^2e^{-(x_1-x_2)^2/4\sigma^2}.
\end{equation}
\begin{figure*}[!ht]
\centering
\includegraphics[width=0.49\linewidth]{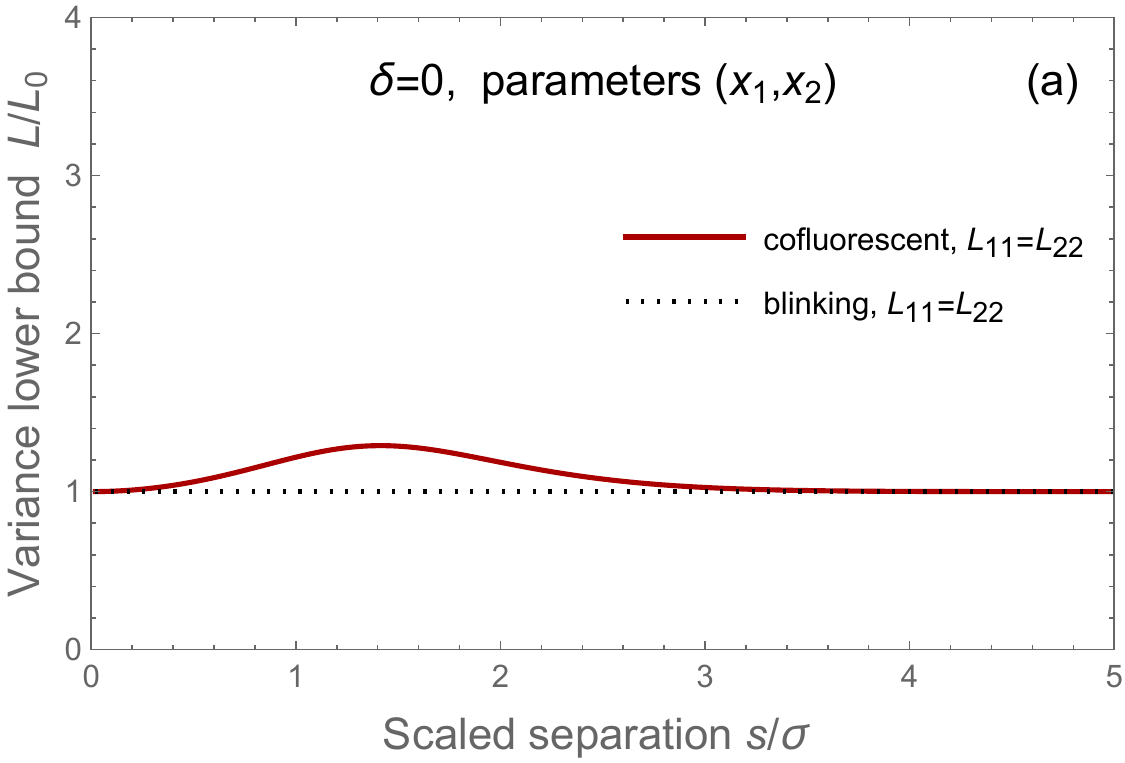}
\includegraphics[width=0.49\linewidth]{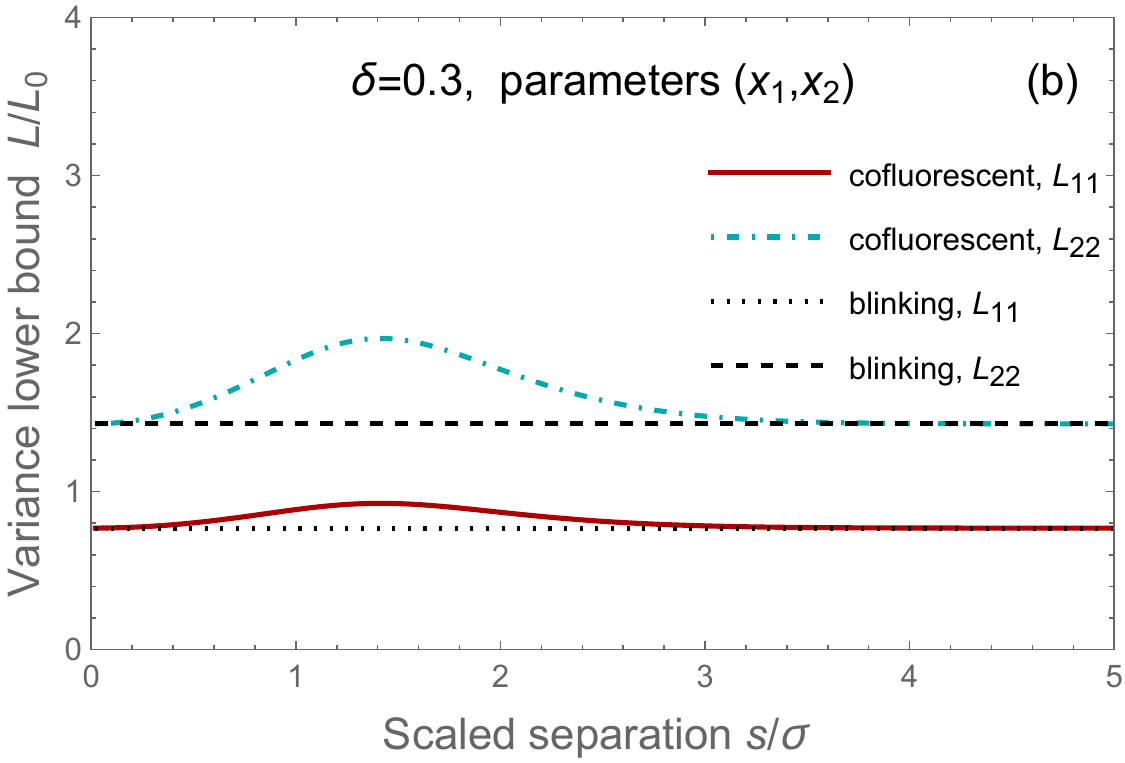}
\includegraphics[width=0.49\linewidth]{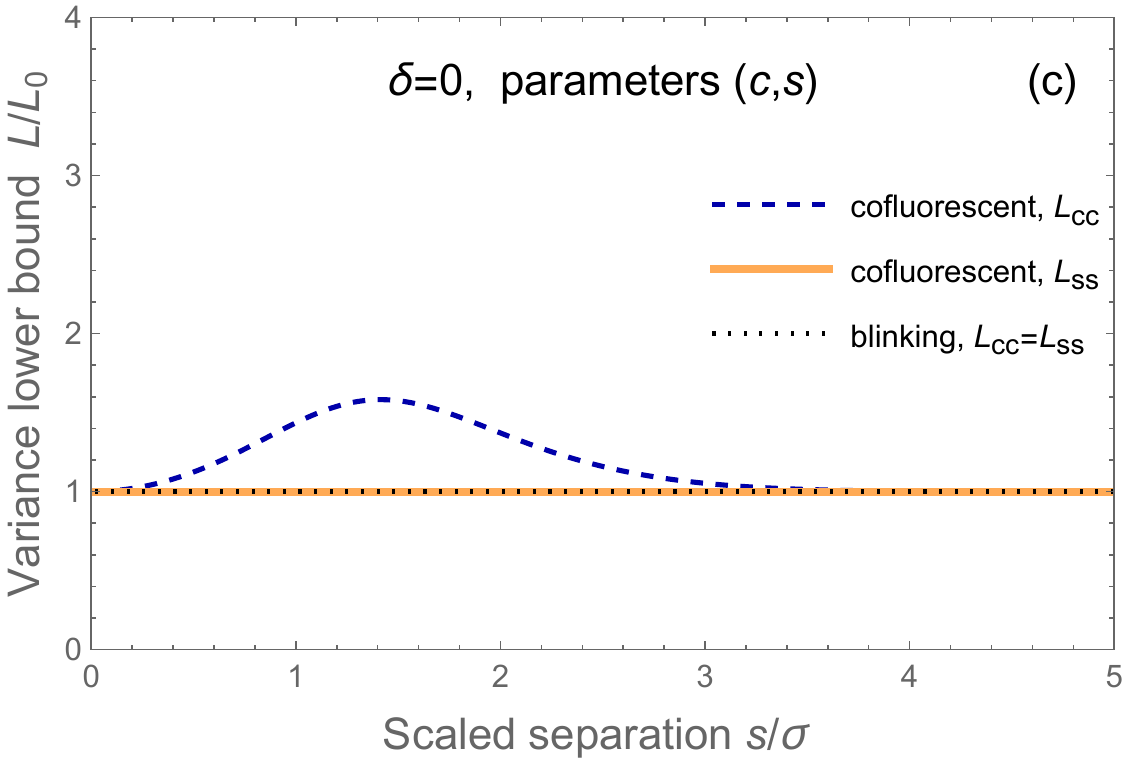}
\includegraphics[width=0.49\linewidth]{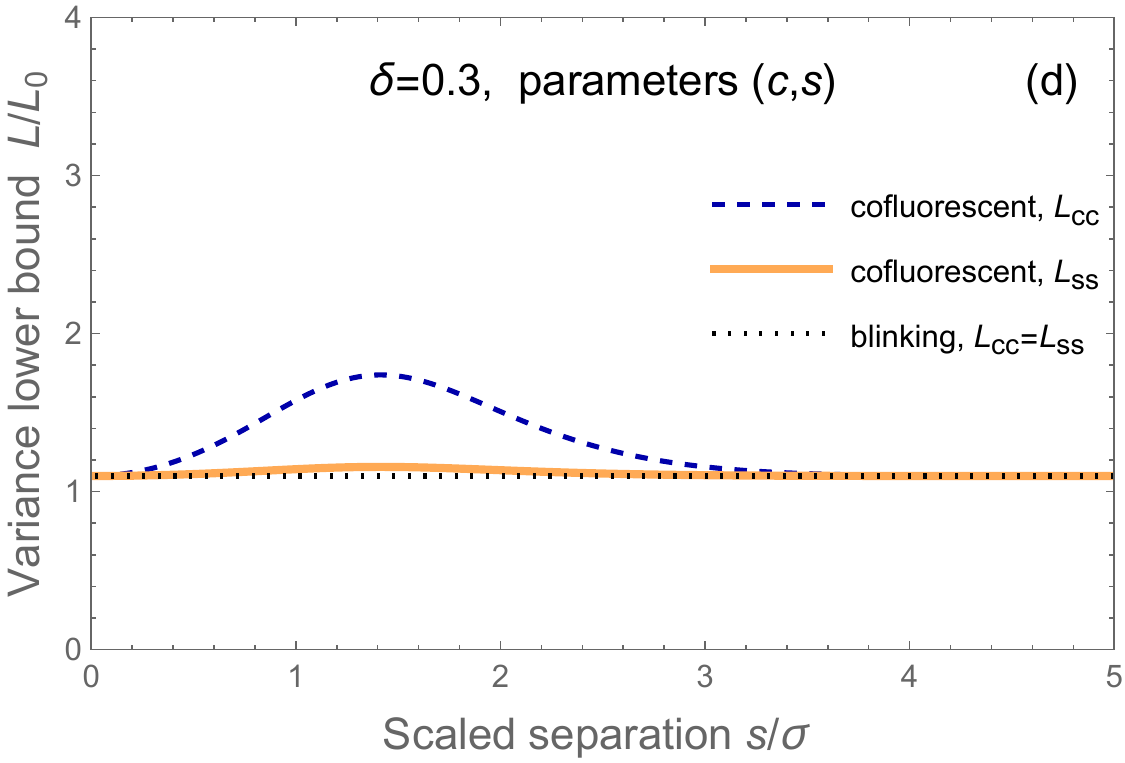}
\includegraphics[width=0.49\linewidth]{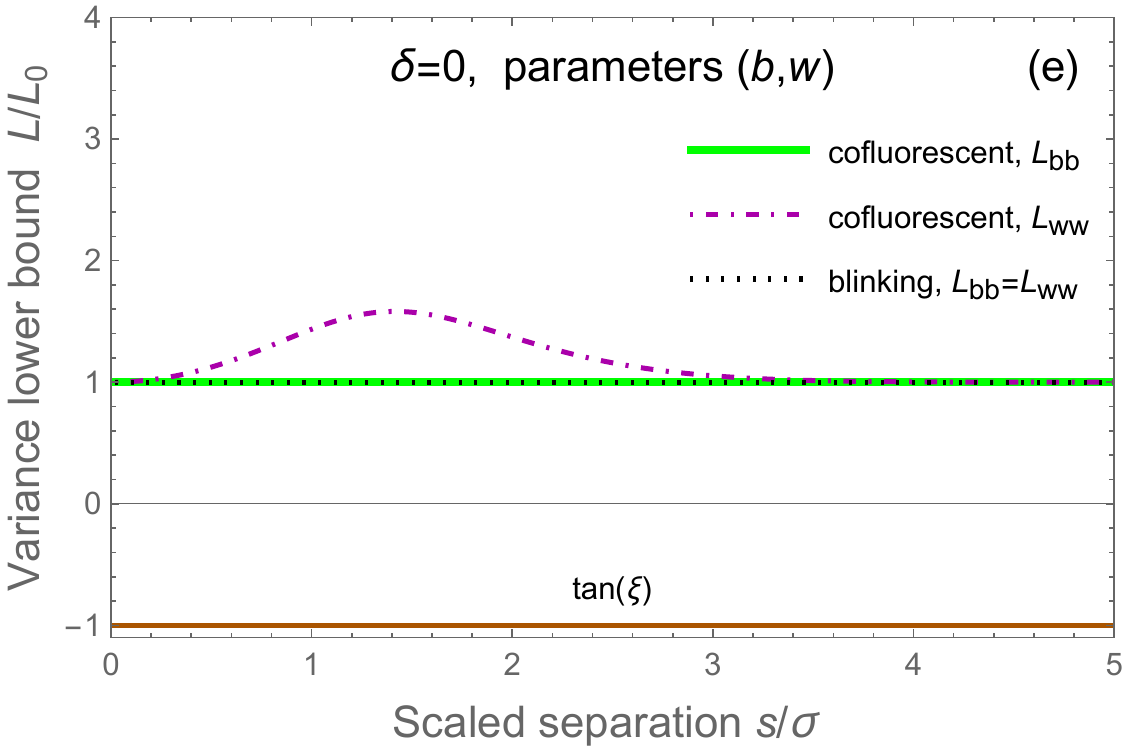}
\includegraphics[width=0.49\linewidth]{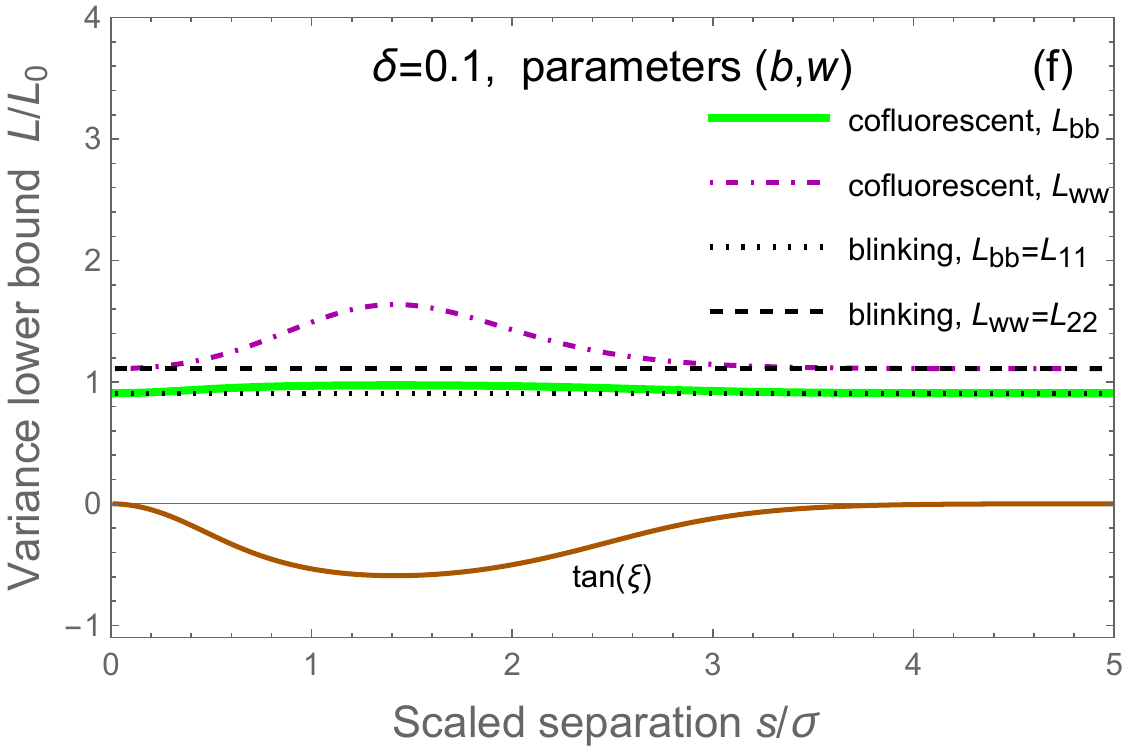}
\caption{Variance lower bounds in the units of $L_0=2\sigma^2/N$ for estimating the unknown positions $(x_1,x_2)$ of two incoherent sources by measuring optimally $N$ photons in a microscope with the point-spread function of standard deviation $\sigma$. $\delta$ is the relative brightness difference of the sources. Parameters $(c,s)$ are the scaled centroid and separation, respectively. Parameters $(b,w)$ are the eigenparameters of the Fisher information matrix obtained from $(x_1,x_2)$ via a rotation by angle $\xi$: $b=x_1\cos\xi + x_2\sin\xi $ and $w= x_2\cos\xi - x_1\sin\xi$. \label{fig:LQ}}
\end{figure*}

The variance lower bounds $L_{ii}=(\dirprod{\mathbf{Q}}^{-1})_{ii}$ are shown in Fig. \ref{fig:LQ}(a) for the case of equal brightnesses $\mu_1=\mu_2=\frac12$ and in Fig. \ref{fig:LQ}(b) for the case of unequal brightnesses $\mu_1\ne\mu_2$. We see that, in contrast to Figs. \ref{fig:L}(a) and (b), the variance lower bound in the cofluorescence scenario does not go to infinity at zero separation. This fact indicates the possibility of an efficient localization of two incoherent sources by means of optimal measurements, such as spatial mode demultiplexing \cite{Tsang16,Paur16,Tsang17,Boucher20,Rouviere24}. In Fig.~\ref{fig:LQ}(a), the lower bound for the cofluorescent scenario lies above that for the blinking scenario, in correspondence with Eq.~(\ref{MainQ}). 

\begin{figure*}[!ht]
\centering
\includegraphics[width=0.49\linewidth]
{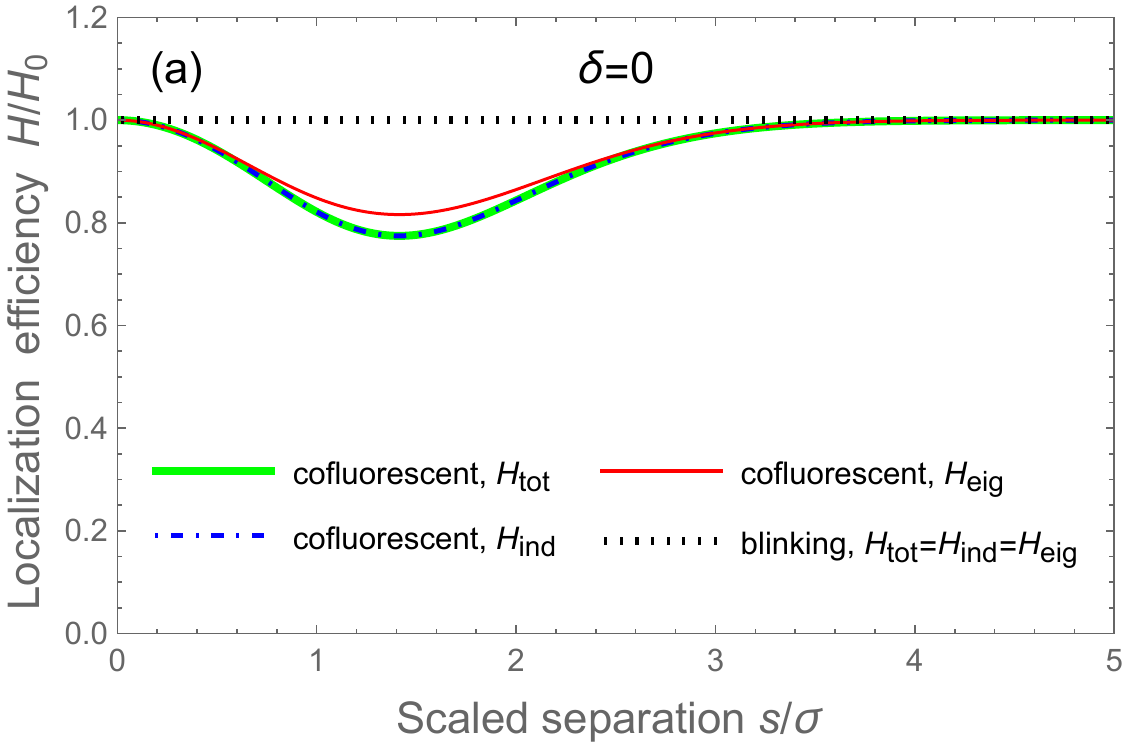}
\includegraphics[width=0.49\linewidth]
{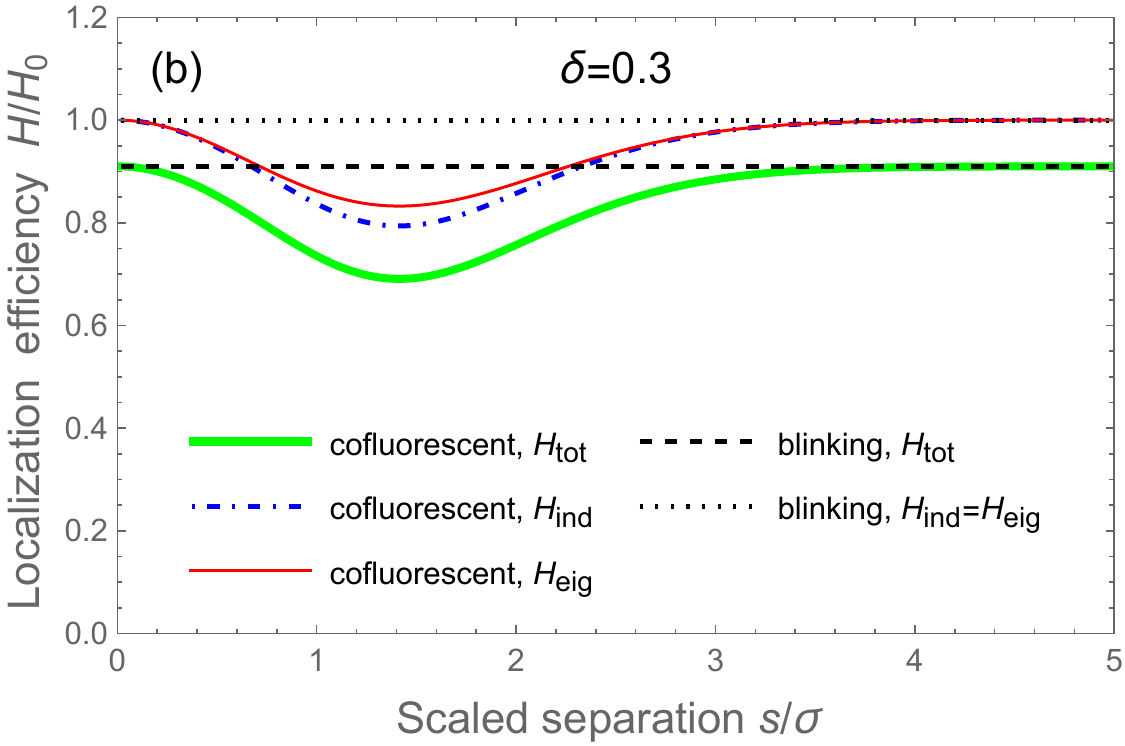}
\caption{Localization efficiency in the units of $H_0=N/2\sigma^2$ for estimating the unknown positions $(x_1,x_2)$ of two incoherent sources by optimally measuring $N$ photons in a microscope with the point-spread function of standard deviation $\sigma$ under condition of saturated quantum Cram\'er-Rao bound. The blinking scenario demonstrates a higher localization efficiency, whatever measure is used for the latter. \label{fig:EfficiencyQ}}
\end{figure*}

A similar behavior is shown in Figs.~\ref{fig:LQ}(c) and (d) for the scaled centroid and separation. The quantum Fisher information matrix is obtained in this basis as $\tilde{\dirprod{\mathbf{Q}}} =\mathbf{O}_{\pi/4}^T\dirprod{\mathbf{Q}}\mathbf{O}_{\pi/4}$, which gives the matrix $\tilde{\dirprod{\mathbf{Q}}}_\text{blink}$ in the form of Eq.~(\ref{tildeFblink1D}) and 
\begin{equation}\label{tildeQcoflu1D}
\tilde{\dirprod{\mathbf{Q}}}_\text{coflu} = \frac{N}{2\sigma^2}\left(\begin{array}{cc}
    1-2\beta & -\delta \\
    -\delta & 1
\end{array}\right).   
\end{equation}

The variance lower bounds for the eigenparameters $b$ and $w$ are shown in Figs.~\ref{fig:LQ}(e) and (f) together with the rotation angle $\xi$. We see that the latter tends to zero at zero separation in the case of unequal brightnesses, meaning that the best parameter in this limit is the position of the brighter source rather than the centroid as in Fig.~\ref{fig:L}(f).

The localization efficiency defined in three possible ways for a saturated quantum Cram\'er-Rao bound is shown in Fig.~\ref{fig:EfficiencyQ}(a) for the case of equal brightnesses. As in Fig.~\ref{fig:EfficiencyF}(a), we see that all three measures of localization efficiency are higher for blinking sources, which is a consequence of Eq.~(\ref{MainQ}). The case of unequal brightnesses is shown in Fig.~\ref{fig:EfficiencyQ}(b). We see that $H_\text{eig}$ is higher for blinking sources than for cofluorescent ones, which is a direct consequence of the trace convexity of the quantum Fisher information matrix, proven in Theorem \ref{theorem:convexityQ}.


\section{Conclusion}
Our comparison of cofluorescent and blinking scenarios for the task of localizing dim incoherent sources of light showed that the blinking scenario always has an informational advantage when the localization efficiency is measured by the average eigenparameter precision, which is a consequence of a fundamental property of the Fisher information matrix, its trace convexity proven in Theorem \ref{theorem:convexity}. This advantage is also proven for a more popular measure of efficiency, average position variance, however, only in the case of translation and rotation invariant microscopy and equal numbers of photons emitted by each source, which is a consequence of another fundamental property of Fisher information matrix, its weak convexity proven in the same theorem. Moreover, we have shown that this advantage is a consequence of the corresponding properties -- trace and weak matrix convexity -- of the quantum Fisher information matrix calculated for light on the image plane of a microscope. The informational advantage of the blinking scenario can be ascribed to a higher \emph{a priori information}, which one has about the number of simultaneously active emitters: 1 in the blinking scenario and $K>1$ in the cofluorescent one. The advantage in a priori information can be quantified by the advantage in the Fisher information calculated in this work.

We believe that these results are important for both the fields of superresolving optical microscopy and quantum metrology, at the junction of which they are located. On the one hand, they show that the information gain produced by blinking, well-known for decades in the superresolving optical microscopy, is a mathematical consequence of such properties of the Fisher information matrix as additivity and convexity. On the other hand, it shows that blinking of the sources can increase the Fisher information on their location even to a higher extent than the optimization of the quantum measurement, much studied in the last years in the framework of quantum metrology. 

The obtained results are important for a deeper understanding of techniques based on blinking fluorophores \cite{Vangindertael18,Schermelleh19,Defienne24}, which are instrumental in advancing superresolution imaging techniques and improving the sensitivity of fluorescence-based biosensors. Both fields are essential for modern biomedical research.

The area of applicability of these results is not reduced to spatial superresolution but also comprises superresolving spectroscopy that allows one to distinguish overlapping spectral lines of optical \cite{Motka16} and magnetic \cite{Rotem19,Oviedo20} fields. Another interesting extension of the developed theory is the analysis of super-resolution optical fluctuation imaging (SOFI) \cite{Dertinger09}, where the brightnesses of the sources do not fall to zero, as in SMLM, but fluctuate in time. A future study will show how the information gain provided by this technique is related to the convexity of the Fisher information.

\section*{Acknowledgments}
This work was supported by QuantERA ERA-NET Cofund programme of the European Union under project EXTRASENS (by the Research Council of Finland, decision 361115, and the Federal Ministry of Education and Research of Germany, grant 13N16935). It was funded by the Research Council of Finland (Flagship Programme PREIN, decision 346518; CHARACTER project, decision 357033), and the Horizon Europe MSCA FLORIN Project 101086142.

\bibliography{BIB-Superres2024}

\end{document}